\documentclass[10pt,a4paper]{amsart}
\usepackage[margin=20mm]{geometry}
\usepackage[numbers]{natbib}
\usepackage{hyperref}

\usepackage{amssymb,amsmath}
\usepackage{amsthm}
\usepackage{bbm}
\usepackage{stmaryrd}
\usepackage[usenames,dvipsnames]{xcolor}
\usepackage{color}

\input xy
\xyoption{all} %\tolerance=500

\numberwithin{equation}{section}

\newtheorem{theorem}{Theorem}[section]

\newtheorem{lemma}[theorem]{Lemma}
\newtheorem{proposition}[theorem]{Proposition}

\theoremstyle{definition}
\newtheorem{definition}[theorem]{Definition}
\newtheorem{remark}[theorem]{Remark}
\newtheorem{example}[theorem]{Example}

\newcommand{\Id}{\mathbbmss{1}}

\newcommand{\rmh}{\textnormal{h}}

\DeclareMathOperator{\Vect}{Vect}

\DeclareMathOperator{\tr}{tr}

\font\black=cmbx10 \font\sblack=cmbx7 \font\ssblack=cmbx5 \font\blackital=cmmib10  \skewchar\blackital='177
\font\sblackital=cmmib7 \skewchar\sblackital='177 \font\ssblackital=cmmib5 \skewchar\ssblackital='177
\font\sanss=cmss10 \font\ssanss=cmss8 %scaled 900
\font\sssanss=cmss8 scaled 600 \font\blackboard=msbm10 \font\sblackboard=msbm7 \font\ssblackboard=msbm5
\font\caligr=eusm10 \font\scaligr=eusm7 \font\sscaligr=eusm5  \font\fraktur=eufm10
\font\sfraktur=eufm7 \font\ssfraktur=eufm5 
\font\bsymb=cmsy10 scaled\magstep2
\def\all#1{\setbox0=\hbox{\lower1.5pt\hbox{\bsymb
       \char"38}}\setbox1=\hbox{$_{#1}$} \box0\lower2pt\box1\;}
\def\exi#1{\setbox0=\hbox{\lower1.5pt\hbox{\bsymb \char"39}}
       \setbox1=\hbox{$_{#1}$} \box0\lower2pt\box1\;}

\def\tx#1{{\fam0\relax#1}}

\newfam\bifam
\textfont\bifam=\blackital \scriptfont\bifam=\sblackital \scriptscriptfont\bifam=\ssblackital

\newfam\blfam
\textfont\blfam=\black \scriptfont\blfam=\sblack \scriptscriptfont\blfam=\ssblack

\newfam\bbfam
\textfont\bbfam=\blackboard \scriptfont\bbfam=\sblackboard \scriptscriptfont\bbfam=\ssblackboard

\newfam\ssfam
\textfont\ssfam=\sanss \scriptfont\ssfam=\ssanss \scriptscriptfont\ssfam=\sssanss
\def\sss#1{{\fam\ssfam\relax#1}}

\newfam\clfam
\textfont\clfam=\caligr \scriptfont\clfam=\scaligr \scriptscriptfont\clfam=\sscaligr

\newfam\frfam
\textfont\frfam=\fraktur \scriptfont\frfam=\sfraktur \scriptscriptfont\frfam=\ssfraktur

\def\hpb#1{\setbox0=\hbox{${#1}$}
    \copy0 \kern-\wd0 \kern.2pt \box0}
\def\vpb#1{\setbox0=\hbox{${#1}$}
    \copy0 \kern-\wd0 \raise.08pt \box0}

\def\pmb#1{\setbox0\hbox{${#1}$} \copy0 \kern-\wd0 \kern.2pt \box0}
\def\pmbb#1{\setbox0\hbox{${#1}$} \copy0 \kern-\wd0
      \kern.2pt \copy0 \kern-\wd0 \kern.2pt \box0}
\def\pmbbb#1{\setbox0\hbox{${#1}$} \copy0 \kern-\wd0
      \kern.2pt \copy0 \kern-\wd0 \kern.2pt
    \copy0 \kern-\wd0 \kern.2pt \box0}
\def\pmxb#1{\setbox0\hbox{${#1}$} \copy0 \kern-\wd0
      \kern.2pt \copy0 \kern-\wd0 \kern.2pt
      \copy0 \kern-\wd0 \kern.2pt \copy0 \kern-\wd0 \kern.2pt \box0}
\def\pmxbb#1{\setbox0\hbox{${#1}$} \copy0 \kern-\wd0 \kern.2pt
      \copy0 \kern-\wd0 \kern.2pt
      \copy0 \kern-\wd0 \kern.2pt \copy0 \kern-\wd0 \kern.2pt
      \copy0 \kern-\wd0 \kern.2pt \box0}

\mathchardef\za="710B  %\alpha
\mathchardef\zb="710C  %\beta
\mathchardef\zg="710D  %\gamma
\mathchardef\zd="710E  %\delta
\mathchardef\zve="710F %\epsilon
\mathchardef\zz="7110  %\zeta
\mathchardef\zh="7111  %\eta
\mathchardef\zvy="7112 %\theta
\mathchardef\zi="7113  %\iota
\mathchardef\zk="7114  %\kappa
\mathchardef\zl="7115  %\lambda
\mathchardef\zm="7116  %\mu
\mathchardef\zn="7117  %\nu
\mathchardef\zx="7118  %\xi
\mathchardef\zp="7119  %\pi
\mathchardef\zr="711A  %\rho
\mathchardef\zs="711B  %\sigma
\mathchardef\zt="711C  %\tau
\mathchardef\zu="711D  %\upsilon
\mathchardef\zvf="711E %\phi
\mathchardef\zq="711F  %\chi
\mathchardef\zc="7120  %\psi
\mathchardef\zw="7121  %\omega
\mathchardef\ze="7122  %\varepsilon
\mathchardef\zy="7123  %\vartheta
\mathchardef\zf="7124  %\varomega
\mathchardef\zvr="7125 %\varrho
\mathchardef\zvs="7126 %\varsigma
\mathchardef\zf="7127  %\varphi
\mathchardef\zG="7000  %\Gamma
\mathchardef\zD="7001  %\Delta
\mathchardef\zY="7002  %\Theta
\mathchardef\zL="7003  %\Lambda
\mathchardef\zX="7004  %\Xi
\mathchardef\zP="7005  %\Pi
\mathchardef\zS="7006  %\Sigma
\mathchardef\zU="7007  %\Upsilon
\mathchardef\zF="7008  %\Phi
\mathchardef\zW="700A  %\Omega
\mathchardef\zC="7009  %\Psi

\newcommand{\be}{\begin{equation}}
\newcommand{\ee}{\end{equation}}

\newcommand{\bea}{\begin{eqnarray}}
\newcommand{\eea}{\end{eqnarray}}
\def\*{{\textstyle *}}
\newcommand{\R}{{\mathbb R}}

\newcommand{\Z}{{\mathbb Z}}

\newcommand{\ot}{\otimes}
\newcommand{\s}{{\textstyle *}}

\def\Vect{\sss{Vect}}

\def\xi{\tx{i}}

\def\xd{\operatorname{d}}

\def\s*{{\scriptstyle *}}

\def\cO{\mathcal{O}}

\newcommand{\beas}{\begin{eqnarray*}}
\newcommand{\eeas}{\end{eqnarray*}}
\def\half{\frac{1}{2}}

\title{Riemannian structures on $\Z_2^n$-manifolds }

   \author{Andrew James Bruce$^\dag$ \& Janusz Grabowski$^\ddag$}
 \address{ $^\dag$Mathematics Research Unit, University of Luxembourg,  Esch-sur-Alzette, Luxembourg \\
   \newline  $^\ddag$ Institute of Mathematics, Polish Academy of Sciences, Poland}
   \email{andrewjamesbruce@googlemail.com,~jagrab@impan.pl}

\date{\today}
 
\begin{document}

\begin{abstract}
Very loosely, $\Z_2^n$-manifolds are `manifolds' with $\Z_2^n$-graded coordinates and their sign rule is determined by the scalar product of their $\Z_2^n$-degrees. A little more carefully, such objects can be understood within a sheaf-theoretical  framework, just as supermanifolds can, but with subtle differences. In this paper, we examine the notion of a Riemannian $\Z_2^n$-manifold, i.e., a $\Z_2^n$-manifold equipped with a Riemannian metric that may carry non-zero $\Z_2^n$-degree. We show that the basic notions and tenets of Riemannian geometry directly generalise to the setting of $\Z_2^n$-geometry. For example, the Fundamental Theorem holds in this higher graded setting.  We  point out the similarities and differences with Riemannian supergeometry. \par
\smallskip\noindent
{\bf Keywords:}
$\Z_2^n$-manifolds;~Riemannian structures;~affine connections.\par
\smallskip\noindent
{\bf MSC 2010:}~14A22;~53B05;~53C50;~58A50. 	
\end{abstract}

 \maketitle

\setcounter{tocdepth}{2}
 \tableofcontents

\section{Introduction}
One can hardly overstate the importance of  (pseudo-)Riemannian geometry in physics, and in particular within metric theories of gravity theories such as general relativity. The notion of a Riemannian metric is also central in geometric mechanics, classical field theories, various non-linear sigma models, string theories, as well as  global analysis, stochastic differential geometry, etc.  Riemannian structures on supermanifolds are not so well studied as their classical counterparts.  Relatively recent papers on the subject include \cite{Galaev:2012,Garnier:2014,Garnier:2012,Goertsches:2008}. One of the interesting aspects of the Riemannian supermanifolds is that we can have metrics that are either degree zero or degree one, referred to as even and odd Riemannian metrics, respectively. Within the context of string theory,  it is now well appreciated that generalised complex geometry is the right geometrical framework to describe NS-NS fluxes. It appears that RR-fluxes  require even Riemannian supergeometry as their geometric framework (see for example \cite{Grassi:2007}). In contrast, odd Riemannian supergeometry remains a mathematical curiosity in the sense that applications in physics have not been so forthcoming.  This is in stark contrast with odd symplectic geometry which is central to the BV-BRST formalism. \par
Recently, there has been a series of papers developing $\Z_2^n$-geometry ($\Z_2^n := \Z_2 \times \cdots \times \Z_2 $) and related algebraic questions, including \cite{Bruce:2020,Bruce:2018,Bruce:2018b,Covolo:2012,Covolo:2016a,Covolo:2016b,Poncin:2016}.  Alongside this, there has been a renewed interest in the physical applications of $\Z_2^n$-gradings, see for example \cite{Aizawa:2020a,Aizawa:2020b,Aizawa:2016,Aizawa:2020c,Aizawa:2020d,Bruce:2019a,Bruce:2020a,Bruce:2019aa,Tolstoy:2014b}. The origin of $\Z_2^n$-Lie algebras and their associative algebra counterparts can be traced back to  Rittenberg \& Wyler \cite{Rittenberg:1978} and Scheunert \cite{Scheunert:1979}, both motivated by physical considerations.  It has long been appreciated that $\Z_2 \times \Z_2$-gradings are important in parastatistics, see for example \cite{Tolstoy:2014} and references therein. However, the r\^{o}le  and importance of these higher gradings in physics are not entirely clear.  That said, $\Z_2^n$-manifolds represent a natural, but challenging generalisation of supermanifolds.  Informally, $\Z_2^n$-manifolds are `manifolds' for which we have coordinates that are $\Z_2^n$-commutative, the sign rule being determined by their $\Z_2^n$-degree and the standard scalar product. Importantly, this means that we have formal coordinates that are not nilpotent, provided $n \geq 2$, naturally, $n=1$ is exactly the theory of supermanifolds.  We view $\Z_2^n$-manifolds, much like supermanifolds, as a ‘pit stop’ on the journey from classical differential geometry to noncommutative geometry.  Higher graded-commutative versions of geometry act as a testbed for ideas in fully-fledged noncommutative geometry.    \par
In this paper, we present the notion of a Riemannian $\Z_2^n$-manifold   and make an initial study of the generalisation of the standard objects found in classical Riemannian geometry, i.e.,  Riemannian metrics, Levi-Civita connections, and the curvature tensors. The bulk of Section \ref{sec:RiemZman} is devoted to developing the general theory of Riemannian $\Z_2^n$-geometry. We begin with a recollection of the aspects of $\Z_2^n$-manifolds that we need for the remainder of this paper.  Like Riemannian supergeometry where we have even and odd metric, on  $\Z_2^n$-manifolds we can have Riemannian metrics that carry non-zero $\Z_2^n$-degree (see Definition \ref{def:RiemMet}). The non-degeneracy of the Riemannian metric places constraints on the dimension of the $\Z_2^n$-manifold depending on the $\Z_2^n$-degree of the metric (see proposition \ref{prop:NonDegCon}).  Most notability, we have a generalisation of the fundamental theory of Riemannian geometry to the setting of $\Z_2^n$-geometry which does not depend on the $\Z_2^n$-degree of the Riemannian metric (see Theorem \ref{trm:FundTheo}). In short, there is a natural generalisation of the Levi-Civita connection in the setting of $\Z_2^n$-geometry. From there we define the Riemann curvature tensor and show that it satisfied the expected generalisation of the Bianchi identities (see Definition \ref{def:RiemCurTen}, Proposition \ref{prop:FirstBianc} and Proposition \ref{prop:SecondBianc}). The Ricci curvature tensor and Ricci scalar can, of course, also be defined (see Definition \ref{def:RicCurv} and Definition \ref{def:RicScal}). We can also classify metrics by their total degree and so we have a natural notion of even and odd Riemannian metrics on $\Z_2^n$-manifolds. For the most part, even and odd Riemannian metrics are not so different: all  the notions and basic constructions are insensitive to the degree.  However, for odd Riemannian metrics, the Ricci scalar identically vanishes (see Proposition \ref{prop:OddRVan}). Moreover, we show that the connection Laplacian (acting on functions) associated with an odd Riemannian metric vanishes (see Theorem \ref{trm:OddLapVanish}). We end this paper in Section \ref{sec:Disc} with a discussion of the consequences of  some of the results  derived in this paper in generalising de Sitter, anti-de Sitter and Einstein manifolds to $\Z_2^n$-geometry. \par
We must remark that the global algebra of functions on a $\Z_2^n$-manifold is an example of an almost commutative algebra in the sense of  Bongaarts \& Pijls \cite{Bongaarts:1994}, which were constructed to give a very workable class of noncommutative geometries.  It is known that a $\Z_2^n$-manifold is fully described by its algebra of global sections (see \cite[Corollary 3.8]{Bruce:2018b}). Linear connections on almost commutative algebras were given by Ciupal\u{a}  \cite{Ciupala:2003}. Moreover, the notion of  a metric and the Levi-Civita connection etc. on an almost commutative algebra were given by Nagakeu \cite{Ngakeu:2007}. While there is some crossover with the results of Nagekeu, the key difference is that we have both global and local algebras of functions as well as coordinate expressions.  Some of the results in this paper represent a generalisation of the discussion of Asorey \& Lavrov \cite{Asorey:2009} comparing even and odd Riemannian metrics on supermanifolds. For a comparison of the Laplacians found in Poisson and Reimannian supergeometry, the reader can consult Khudaverdian \& Voronov \cite{Khudaverdian:2002}.

\section{Riemannian $\Z_2^n$-manifolds}\label{sec:RiemZman}
\subsection{The basics of $\Z_2^n$-geometry}
The locally ringed space approach to $\Z_2^n$-manifolds was originally given by Covolo, Grabowski and Poncin in   \cite{Covolo:2016a,Covolo:2016b}. Technical issues related to the functional analytic properties of the structure sheaves of $\Z_2^n$-manifolds and their products were carefully explored in \cite{Bruce:2018,Bruce:2018b}.  We will draw upon these works and not present proofs of any formal statements. To set some initial notation, by $\Z_2^n$ we mean the abelian group $\Z_2 \times \Z_2 \times \cdots \times \Z_2$ where the Cartesian product is over $n$ factors. With an ordering fixed, we will often denote elements of $\Z_2^n$ as $\gamma_i$ for $i = 0 ,1 \cdots , N$, where  $N = 2^n-1$.  We will use the convention of ordering the elements of $\Z_2^n$  by filling in the zeros and the ones from the left,  and then  placing the even elements first and then the odd elements. For example, with this choice of ordering
\begin{align*}
& \Z_{2}^{2} = \{ (0,0),  \: (1,1), \: (0,1), \: (1,0)\},\\
& \Z_2^3 = \{ (0, 0, 0), \:  (0, 1, 1), \: (1, 0, 1), \: (1, 1, 0),\:  (0, 0, 1), \: (0, 1, 0), \: (1, 0, 0),  \:  (1, 1, 1) \}.
\end{align*}
Naturally, other choices of ordering are possible and have been used in the literature. Throughout this paper we set $\mathbf{q} = (q_{1}, q_{2}, \cdots , q_{N})$, where  $N = 2^n-1$ and $q_i  \in \mathbb{N}$.
\begin{definition}
A \emph{locally} $\Z_{2}^{n}$-\emph{ringed space}, $n \in \mathbb{N} \setminus \{0\}$, is a pair $S := (|S|, \cO_S )$ where $|S|$ is a second-countable Hausdorff space, and a $\cO_{S}$  is a sheaf  of $\Z_{2}^{n}$-graded $\Z_{2}^{n}$-commutative associative unital $\R$-algebras, such that the stalks $\cO_{S,p}$, $p \in  |S|$ are local rings.
\end{definition}
In this context,  $\Z_2^n$ -commutative means that any two sections $s$, $t \in \cO_{S}(|U|)$, $|U| \subset |S|$ open, of homogeneous degree $\deg(s)  \in \Z_{2}^{n}$ and $\deg(t) \in \Z_2^n$ commute up to the sign rule
$$st = (-1)^{\langle \deg(s), \deg(t)\rangle} \: ts,$$
where $\langle \: , \:\rangle$ is the standard scalar product on $\Z_2^n$. For example, consider the  $\gamma_3 =(0,1)$ and $\gamma_2 =(1,1) \in \Z_2^2$, then $\langle \gamma_3, \gamma_2 \rangle  = 0 \times 1 + 1 \times 1 =1$. If we write `symmetric' or `skew-symmetric', we will always mean in the $\Z_2^n$-graded sense unless otherwise stated.
\begin{definition}\label{def:Z2nGrassmann}
 A  $\Z_2^n$-\emph{Grassmann algebra}  $\Lambda^{\mathbf{q}} := \R[[\zx]]$ is a formal power series generated by the $\Z_2^n$-graded variables $\{ \zx^A \}$, there are $q_i$ generators of degree $\gamma_i \in \Z_2^n$,  ($i >0$), subject to the  relation
$$\zx^A \zx^B = (-1)^{\langle \deg(A) , \deg(B)  \rangle } \zx^B \zx^A,$$
where $\deg(\zx^A) =: \deg(A) \in \Z_2^n\setminus \{ 0\}$ and similar for $\zx^B$.
\end{definition}
It is important to note that elements of such algebras are formal power series  and not, in general, polynomials as we have non-nilpotent generators, i.e., the generators of even total degree.  A $\Z_2^n$-manifold is, very loosely, a manifold whose structure sheaf is modified to include generators of a $\Z_2^n$-Grassmann algebra.
\begin{definition}
A (smooth) $\Z_2^n$-\emph{manifold} of dimension $p |\mathbf{q}$ is a locally $\Z_{2}^{n}$-ringed space $ M := \left(|M|, \cO_M \right)$, which is locally isomorphic to the $\Z_2^n$-ringed space $\mathbb{R}^{p |\mathbf{q}} := \left( \mathbb{R}^{p}, C^{\infty}_{\mathbb{R}^{p}}(-)[[\zx]] \right)$. Local sections of $M$ are formal power series in the $\Z_2^n$-graded variables $\zx$ with  smooth coefficients,
$$\cO_M(|U|) \simeq C^{\infty}(|U|)[[\zx]] :=  \left \{ \sum_{\alpha \in \mathbb{N}^{N}}^{\infty}  \zx^{\alpha}f_{\alpha} ~ | \: f_{ \alpha} \in C^{\infty}(|U|)\right \},$$
for `small enough' open domains $|U|\subset |M|$.   \emph{Morphisms} between $\Z_{2}^{n}$-manifolds are  morphisms of $\Z_{2}^{n}$-ringed spaces, that is,  pairs $(\phi, \phi^{*}) : (|M|, \cO_M) \rightarrow  (|N|, \cO_N)$ consisting of  of a continuous map  $\phi: |M| \rightarrow |N|$ and sheaf morphism $\phi^{*}: \cO_N(|V|) \rightarrow \cO_M(\phi^{-1}(|V|))$, where $|V| \subset |N|$  is open. We will refer to the global sections of the structure sheaf $\cO_M$ as \emph{functions} on $M$ and denote then as $C^{\infty}(M) := \cO_{M}(|M|)$.
\end{definition}
\begin{example}[The local model]\label{exp:SuperDom}
The locally $\Z_{2}^{n}$-ringed space $\mathcal{U}^{p|\mathbf{q}} :=  \big(\mathcal{U}^p , C^\infty_{\mathcal{U}^p}(-)[[\zx]] \big)$, where $\mathcal{U}^p \subseteq \R^p$ is naturally a $\Z_2^n$-manifold -- we refer to such $\Z_2^n$-manifolds as \emph{$\Z_2^n$-domains} of dimension $p|\mathbf{q}$.  We can employ (natural) coordinates $x^I := (x^a, \zx^A)$ on any $\Z_2^n$-domain, where $x^a$ form a coordinate system on $\mathcal{U}^p$ and the $\zx^A$ are formal coordinates.
\end{example}
Associated with any $\Z_{2}^{n}$-graded algebra $\mathcal{A}$ is the  ideal $J$ of $\mathcal{A}$ generated by all homogeneous elements of $\mathcal{A}$ that have non-zero degree. The associated $J$-adic topology plays an important r\^ole in the theory of $\Z_{2}^{n}$-manifolds. Given a morphism of $\Z_{2}^{n}$-graded algebras $f : \mathcal{A} \rightarrow \mathcal{A}^{\prime}$, then $f (J_{\mathcal{A}} ) \subset J_{\mathcal{A}^{\prime}}$.  These notions can be `sheafified', i.e., for any $\Z_{2}^{n}$-manifold $M$, there exists an \emph{ideal sheaf} $\mathcal{J}_M$, defined by
$\mathcal{J}(|U| ) := \langle f \in \cO_M(|U|)~|~ \deg(f)\neq 0 \rangle$. The $\mathcal{J}_M$-adic topology on $\cO_M$ can then be defined in the obvious way. Furthermore, for any $\Z_2^n$-manifold~$M$, there exists a short exact sequence of sheaves of $\Z_2^n$-graded $\Z_2^n$-commutative associative $\R$-algebras
\begin{equation}\label{eqn:SES}
0\longrightarrow\ker\zve\longrightarrow\cO_M\stackrel{\zve}{\longrightarrow}C^\infty_{|M|}\longrightarrow 0,
\end{equation}
such that $\ker \zve=\mathcal{J}_M$. Informally, $\zve_{|U|}: \cO_M(|U|) \rightarrow C^\infty(|U|)$ is simply ``throwing away" the formal coordinates.\par
Much like the theory of manifolds, or indeed supermanifolds, one can global geometric concepts via the glueing of local geometric concepts. That is, we can consider a $\Z_2^n$-manifold as being cover by  $\Z_2^n$-domains together with the appropriate glueing information, i.e., admissible coordinate transformations.   Fundamental here is the \emph{chart theorem} (\cite[Theorem 7.10]{Covolo:2016}) that allows us to write morphisms of $\Z_2^n$-manifolds in terms of the local coordinates.  Specifically, suppose we have two  $\Z_2^n$-domains $\mathcal{U}^{p|\mathbf{q}}$ and $\mathcal{V}^{r|\mathbf{s}}$. Then morphisms $\phi: \mathcal{U}^{p|\mathbf{q}} \longrightarrow \mathcal{V}^{r|\mathbf{s}}$  correspond  to \emph{graded unital $\R$-algebra morphisms}
 \begin{equation*}
 \phi^* : C^{\infty}\big(\mathcal{V}^r \big)[[\eta]] \longrightarrow   C^{\infty}\big(\mathcal{U}^p \big)[[\zx]],
 \end{equation*}
which are fully determined by their coordinate expressions.  For any $M := \big( |M|, \cO_M \big)$, we  define open $\Z_2^n$-submanifolds  of $M$ as  $\Z_2^n$-manifolds of the form $U := \big(|U|, \cO_M|_{|U|} \big)$, where $|U|\subset |M|$ is open. Due to the local structure of a  $\Z_2^n$-manifold  of dimension $p|\mathbf{q}$ we know that for a `small enough' $|U| \subset |M|$ there exists an isomorphism of $\Z_2^n$-manifolds
$$\rmh : U \longrightarrow  \mathcal{U}^{p|\mathbf{q}}.$$
This isomorphism allows us to employ local coordinates. A pair $(U, \rmh)$ we refer to as a \emph{(coordinate) chart}, and a family of charts $\{(U_i , \rmh_i)  \}_{i \in \mathcal{I}}$ we refer to as an \emph{atlas}  if the family $\{ |U_i| \}_{i \in \mathcal{I}}$ forms an open cover of $|M|$.  The local structure of a  $\Z_2^n$-manifold  guarantees that an atlas always exists. As a result of the local structure of a $\Z_{2}^{n}$-manifold any morphism $\phi : M \rightarrow N$  can be uniquely specified by a family of local morphisms between $\Z_{2}^{n}$-domains, once atlases on $M$ and $N$ have been fixed. Thus, morphisms of $\Z_{2}^{n}$-manifolds can be fully described using local coordinates. We will regularly exploit this and  employ the standard abuses of notation as found in classical differential geometry when writing morphisms using local coordinates.
\begin{remark}
There is an analogue of the Batchelor--Gaw\c{e}dzki theorem for (real)  $\Z_{2}^{n}$-manifolds,  see \cite[Theorem 3.2]{Covolo:2016a}. That is,  any $\Z_2^n$-manifold is noncanonically isomorphic to a $\Z_2^n\setminus \{ 0\}$-graded vector bundle over  a smooth manifold.  We will not use this theorem in this paper.
\end{remark}
\subsection{Riemannian structures}
For notion in classical Riemannian geometry, we refer the reader to Eisenhart \cite{Eisenhart:1997}. The \emph{tangent sheaf} $\mathcal{T}M$ of a $\Z_2^n$-manifold $M$ is  defined as the sheaf of derivations of sections of the structure sheaf, i.e., $\mathcal{T}M(|U|) := \textnormal{Der}(\cO_M(|U|))$, for arbitrary $|U| \subset |M|$. Naturally, this is  a sheaf of locally free $\cO_{M}$-modules. Global sections of the tangent sheaf are referred  to as \emph{vector fields}. We denote the $\cO_{M}(|M|)$-module of vector fields as $\Vect(M)$.  The dual of the tangent sheaf is the \emph{cotangent sheaf}, which we denote as $\mathcal{T}^*M$. This is also a sheaf of locally free  $\cO_{M}$-modules. Global section of the tangent sheaf we will refer to as \emph{one-forms} and we denote the $\cO_{M}(|M|)$-module of one-forms as $\Omega^1(M)$. For more details about the tangent and cotangent sheaves, the reader should consult \cite{Covolo:2016}.
\begin{definition}\label{def:RiemMet}
A  \emph{Riemannian metric} on a $\Z_2^n$-manifold $M$ is a $\Z_2^n$-homogeneous, $\Z_2^n$-symmetric, non-degenerate, $\cO_M$-linear morphisms of sheaves
$$\langle - | - \rangle_g : ~\mathcal{T}M \otimes_{\cO_M} \mathcal{T}M \longrightarrow \cO_M.$$
A $\Z_2^n$-manifold equipped with a Riemannian metric is referred to as a \emph{Riemannian $\Z_2^n$-manifold}.
\end{definition}
\noindent We will insist that the Riemannian metric is homogeneous with respect to the $\Z_2^n$-degree, and  we will denote the degree of the metric as $\deg(g) \in \Z_2^n$. We will largely formulate  Riemannian geometry globally using vector fields.  \\
\noindent Explicitly, a Riemannian metric has the following properties:
\begin{enumerate}
\setlength\itemsep{0.5em}
\item$\deg(\langle X| Y \rangle_g) = \deg(X) + \deg(Y) + \deg(g)$,\label{MetProp1}
\item $\langle X| Y \rangle_g = (-1)^{\langle\deg(X) , \deg(Y) \rangle} \langle Y| X \rangle_g$,\label{MetProp2}
\item If $\langle X| Y \rangle_g =0$  for all $Y \in \Vect(M)$, then $X =0$,\label{MetProp3}
\item $\langle f X + Y| Z \rangle_g = f \langle X|Z \rangle_g + \langle Y| Z\rangle_g$, \label{MetProp4}
\end{enumerate}
for arbitrary  (homogeneous) $X,Y,Z \in \Vect(M)$ and $f \in C^\infty(M)$.  We will say that a Riemannian metric is \emph{even} if and only if it has total degree zero. Similarly, we will say that a Riemannian metric is \emph{odd} if and only if it has total degree one. Any Riemannian metric we consider will be either even or odd as we will only be considering homogeneous metrics.
\begin{example} Pseudo-Riemannian and Riemannian manifolds are  Riemannian $\Z_2^0$-manifolds. Naturally, there are only even Riemannian metrics in the classical setting.  Similarly,  Riemannian supermanifolds either with an even or odd Riemannian metric are Riemannian $\Z_2$-manifolds.
\end{example}
A Riemannian metric is completely specified by a symmetric rank-two covariant tensor $g = \delta x^I \delta x^J g_{JI}(x)$, where  the components are given by $g_{IJ} :=\langle \partial_I| \partial_J \rangle_g $. Under changes of local coordinates $x^I \mapsto x^{I'}(x)$ the components of the metric transform as
\begin{equation}\label{eqn:TransMet}
g_{I'J'} = (-1)^{\langle \deg(J') , \deg(I) \rangle} \left(\frac{\partial x^I}{\partial x^{I'}}\right) \left(\frac{\partial x^J}{\partial x^{J'}}\right)~  g_{JI}\,,
\end{equation}
where we have  used the symmetry $g_{IJ} = (-1)^{\langle \deg{I}, \deg(J) \rangle} \, g_{JI}$.   The Riemannian metric can then locally be written as
\begin{equation}\label{eqn:LocMet}
 \langle X| Y \rangle_g =  (-1)^{\langle \deg(Y) ,  \deg(I) \rangle} ~ X^I(x)Y^J(x) g_{JI}(x)\,,
 \end{equation}
where  (locally)  $X = X^I \frac{\partial}{\partial x^I}$ and similar for $Y= Y^J \frac{\partial}{\partial x^J}$. A direct computation will show that the above local expression for the metric invariant under changes of coordinates. As customary, we will refer to $g$ as the \emph{metric tensor} or by minor abuse of language, we may refer to $g$ itself as the Riemannian metric. Properties \eqref{MetProp1}, \eqref{MetProp2} and \eqref{MetProp4} of the Riemannian metric are easily seen to hold for the local expression \eqref{eqn:LocMet}. Just as for Riemannian supermanifolds, the non-degeneracy condition places constraints on the number of ``formal directions''.   In particular, the dimension of a supermanifold must be $n|2m$ if it is to admit an even Riemannian metric, and the dimension must be  $n|n$ if it is to admit an odd Riemannian metric. In the continuation of this paper, we will denote a Riemannian $\Z_2^n$-manifold by a pair $(M, g)$.
\begin{proposition}\label{prop:NonDegCon}
Let $(M,g)$ be a Riemannian $\Z_2^n$-manifold. We set $p|\mathbf{q} = p|q_1, q_2, \cdots, q_N =: q_0,q_1, q_2, \cdots , q_N$. Then the non-degeneracy condition \eqref{MetProp3} requires that
\renewcommand\labelenumi{(\roman{enumi})}
\begin{enumerate}
\item $q_i = q_j$ when $\gamma_i + \gamma_j = \deg(g)$,
\item If in addition to the above, $\langle \gamma_i , \gamma_j   \rangle =1$, then $q_i$ (and so $q_j$) must be even integers, i.e., $q_i = 2 n_i$ for some integer $n_i$.
\end{enumerate}
\end{proposition}
\begin{proof}
As this is a local question, it is sufficient to examine the invertibility of the metric tensor $g_{IJ}$ in  some chosen set of local coordinates.  The $i,j$ block of $g$ is of $\Z_2^n$-degree $\gamma_i + \gamma_j + \deg(g)$. The entries of this matrix are in $\cO_M(|U|)$ for some ``small enough'' $|U| \subset |M|$. It is known that any such matrix is invertible if and only if the underlying real matrix $\epsilon_{|U|}(g)$ is itself invertible (see \cite{Covolo:2012} and for the super-case see \cite{Leites:1980}). Here $\epsilon : \cO_M  \rightarrow C^\infty_{|M|}$ is the canonical sheaf morphisms.  Under this projection, the only non-zero blocks are those for which $\gamma_i + \gamma_j = \deg(g)$.
\renewcommand\labelenumi{(\roman{enumi})}
\begin{enumerate}
\item Each block itself must be invertible and so a square matrix. Thus $q_i = q_j$ whenever $\gamma_i + \gamma_j = \deg(g)$.
\item The symmetry of the metric tensor implies that if $\langle \gamma_i, \gamma_i\rangle =1$, these  blocks are, in classical terminology,  skew-symmetric matrices. This directly implies that the $q_i$ must be even.
\end{enumerate}
\end{proof}
\begin{lemma}\label{lem:VecOneform}
Let $(M,g)$ be a Riemannian $\Z_2^n$-manifold. Then there is a one-to-one correspondence between elements of $\Vect(M)$ and $\Omega^1(M)$, where a one-form $\omega$ corresponds to a vector field $X$ if and only if
$$\omega(Z) = \langle Z |X \rangle_g,$$
for all $Z \in \Vect(M)$.
\end{lemma}
 \begin{proof}
 This directly follows from the non-degeneracy of the Riemannian metric.  In local coordinates  we have
 $$\omega(Z) = Z^I \omega_I = (-1)^{\langle \deg(X), \deg(I)  \rangle} ~ Z^I X^J g_{JI},$$
 and so $\omega_I = (-1)^{\langle \deg(X), \deg(I)  \rangle} ~ X^J g_{JI}$.
 \end{proof}
 Note that $\deg(\omega) = \deg(X) + \deg(g)$ and so the above correspondence does not preserve the $\Z_2^n$-degree unless the metric is degree zero.
 \begin{proposition}\label{prop:RedMet}
 Let $(M,g)$ be a Riemannian $\Z_2^n$-manifold with a degree zero Riemannian metric, i.e., $\deg(g) = (0,0,\cdots , 0)$.  Then the reduced manifold $|M|$ comes naturally equipped with an induced pseudo-Riemannian structure $|g|$.
 \end{proposition}
\begin{proof}
Let $\bar{X}$ and $\bar{Y} \in \Vect(|M|)$ be vector fields on $|M|$. We then define the reduced pseudo-Riemannian metric as
$$\langle \bar{X} | \bar{Y} \rangle_{|g|} = \epsilon_{|M|}\big( \langle X| Y  \rangle_g \big)\,,$$
for any degree zero $X$ and $Y \in \Vect(M)$ such that $\bar{X} = X \circ \epsilon_{|M|}$ and $\bar{Y} = Y \circ \epsilon_{|M|}$ (see \eqref{eqn:SES}). This construction gives a pseudo-Riemannian metric on $|M|$ as all the required properties are inherited  from that of the original Riemannian metric. Specifically, with respect to some chosen local coordinates we have
$$|g|_{ab} = \epsilon_{|U|}g_{ab},$$
here $|U| \in |M|$ is chosen `small enough' to allow the use of local coordinates. It is clear that as $g_{IJ}$ is an invertible degree zero matrix, $|g|_{ab}$ is itself invertible and so non-zero.  The symmetry $|g|_{ab} = |g|_{ba}$ is also evident.
\end{proof}
\begin{remark}
Proposition  \ref{prop:RedMet} does \emph{not} generalise to non-zero $\Z_2^n$-degree Riemannian metrics.  For example, an even Riemannian metric on a supermanifold induces a (pseudo-)Riemannian on the reduced manifold via the above construction, while an odd Riemannian metric cannot induce such a structure on the reduced manifold.
\end{remark}
\begin{example}
Consider the  $\Z_2^2$-domain $\R^{2|1,2,2}$ equipped with global coordinates  $$(\underbrace{x^a}_{(0,0)} ,\,  \underbrace{z}_{(1,1)},\,  \underbrace{\zx^b}_{(0,1)}, \, \underbrace{\eta^c}_{(1,0)} ),$$
 where $a,b$ and $c =1,2$.  This $\Z_2^2$-manifold can can be equipped with a degree $(0,0)$  Riemannian metric
$$g = (\delta x^1)^2 + (\delta x^2)^2 + (\delta z)^2 + 2 \, \delta \zx^1 \delta \zx^2 +2 \, \delta \eta^1 \delta \eta^2. $$
The reduced Riemanian metric on the reduced manifold is the standard Euclidean metric on $\R^2$. 
\end{example}
\begin{example}
Recall that the open disk is defined as $\mathbb{D} := \{ (x^1, x^2) \in \R^2~ | ~ (x^1)^2 + (x^2)^2 < 1\}$. We define the $\Z_2^2$-disk as 
$$\mathbb{D}^{2|1,2,2} := \big( \mathbb{D}, ~~ C^\infty_{\mathbb{D}}(-)\widehat{\otimes}_\R ~\R[[z, \zx^a, \eta^b]] \big),$$
where $a$ and $b = 1,2$. The algebras here are nuclear Fr\'{e}chet algebras, and as such the tensor product is completed with respect to any reasonable  topology, for example  the projective topology. By construction the reduced manifold is the open disk.  We can equip this $\Z_2^2$-manifold with a degree $(0,0)$ Riemannian metric given by 
$$g = 4 \, \frac{(\delta x^1)^2 + (\delta x^2)^2 + (\delta z)^2 + 2 \, \delta \zx^1 \delta \zx^2 +2 \, \delta \eta^1 \delta \eta^2}{(1- (x^1)^2 - (x^2)^2   - z^2-\zx^1 \zx^2 - \eta^1 \eta^2 )^2}. $$
Note that the denominator has non-zero real part and so is invertible. The reduced Riemannian metric is the metric of the Poncar\'{e}  disk model. 
\end{example}

\subsection{The Cartesian and warped product of Riemannian $\Z_2^n$-manifolds}
It is known that the category of $\Z_2^n$-manifolds admit all finite products, for details the reader should consult \cite{Bruce:2018b}.  In particular, given two $\Z_2^n$-manifolds $M_1$ and $M_2$, their Cartesian product $M_1 \times M_2$ is well-defined and is itself a $\Z_2^n$-manifold. One can construct  an atlas on $M_1 \times M_2$ using atlases on $M_1$ and $M_2$ in more-or-less the same way as one can in the category of supermanifolds.   \par
We will denote the projections as $\pi_i : M_1 \times M_2 \rightarrow M_i$ ($i = 1,2$). The underlying continuous map $|\pi_i| : |M_1| \times| M_2| \rightarrow |M_i| $ is the standard smooth projection.  Using a basis of the topology for both $M_1$ and $M_2$ consisting of small enough opens $|U_i| \subset |M_i|$, such that the corresponding open $\Z_2^n$-submanifolds are diffeomorphic to coordinate domains, we see that
$$\pi^*_{i, |U_i|} : \cO_{M_i}(|U_i|): \cO_{M_1}(|U_1|) ~\widehat{\otimes}_{\R}~\cO_{M_2}(|U_2|) $$
is the standard inclusion of algebras.  As the algebras in question are nuclear Fr\'{e}chet, the tensor product is completed with respect to any reasonable  topology, say for concreteness, the projective topology (for details of the functional analytic properties of the structure sheaf of a $\Z_2^n$-manifold, the reader should consult \cite{Bruce:2018}).
\begin{proposition}
Let $(M_i, g_i)$ $(i = 1,2)$ be Riemannian $\Z_2^n$-manifolds whose Riemannian metric are of the same $\Z_2^n$-degree. Then,
$$(M_1 \times M_2, g := \pi^*_1 g_1 + \pi^*_2 g_2 )$$
is itself a  Riemannian $\Z_2^n$-manifold whose Riemannian metric is of degree $\deg(g) = \deg(g_i)$.
\end{proposition}
\begin{proof}
It is clear that this construction means that
$$\langle X | Y \rangle_g  = \langle X | Y \rangle_{\pi^*_1g_1} + \langle X | Y \rangle_{\pi^*_2g_2},$$
for all $X$ and $Y \in \Vect(M_1 \times M_2)$.  We then need to check that this has the right properties. It is clear that we have the required $\Z_2^n$-degree, that this is symmetric and is linear.  The non-degeneracy also directly follows as the initial structures are themselves non-degenerate.
\end{proof}
By employing local coordinates $x^I$ on $M_1$ and $y^A$ on $M_2$, the Riemannian metric is given locally by
$$\langle X | Y \rangle_g = (-1)^{\langle \deg(Y) ,  \deg(I) \rangle} ~ X^I(x,y)Y^J(x,y) g_{JI}(x) + (-1)^{\langle \deg(Y) ,  \deg(A) \rangle} ~ X^A(x,y)Y^B(x,y) g_{BA}(y)\,,$$
or as a tensor
$$g =  \delta x^I \delta x^J g_{JI}(x) +  \delta y^A \delta y^B g_{BA}(y).$$
This construction generalises to a warped product.
\begin{definition}
Let $(M_i, g_i)$ $(i = 1, 2)$ be Riemannian $\Z_2^n$-manifolds whose Riemannian metric are of the same $\Z_2^n$-degree. Let $\mu \in C^\infty(M_1)$ be an invertible global function that is strictly positive, i.e. $\epsilon_{M_1}\mu$ is a strictly positive function on $|M_1|$. Then the \emph{warped product} is defined as
$$M_1 \times_\mu M_2 := \big( M_1 \times M_2 ,  ~  g := \pi^*_1 g_1 + (\pi^*_1 \mu) \, \pi^*_2 g_2 \big ).$$
\end{definition}
Note that if $\mu$ is the constant function of value one, then the warped product is simply the Cartesian product.
\begin{proposition}
The warped product $M_1 \times_\mu M_2$ is a Riemannian $\Z_2^n$-manifold.
\end{proposition}
\begin{proof}
It is clear that the definition of the warped product means that
$$\langle X | Y \rangle_g  = \langle X | Y \rangle_{\pi^*_1g_1} + (\pi^*_1\mu) \,\langle X | Y \rangle_{\pi^*_2g_2},$$
for all $X$ and $Y \in \Vect(M_1 \times M_2)$.  We then need to check that this has the right properties. It is clear that we have the required $\Z_2^n$-degree, i.e., the Riemannian metric is homogeneous.  The symmetry  and  linearity directly from the properties of the starting pair of Riemannian metrics.  Note that $\pi^*_1\mu$, as it is invertible has a non-zero real component. Thus the second term in the sum will not vanish unless $\langle X | Y \rangle_{\pi^*_2g_2}$  vanishes. Thus, we are back to the non-degeneracy condition of the pair of Riemannian metrics.
\end{proof}
\begin{example}\label{exp:WarpedMetric}
It is easy to observe that $\R^{1|1,1,1} =\R^{1|1,0,0} \times \R^{0|0,1,1}$. Moreover, in global coordinates $(x, z, \zx, \eta)$ the canonical degree $(1,1)$ Riemannian metric is $g_0 = \delta x\, \delta z + \delta \eta\, \delta \zx$.  Note that this is, of course, the sum of the two Riemannian metric on each factor, i.e., we have a Cartesian product of Riemannian $\Z_2^2$-manifold. However, we can warp this construction.  We  choose,  somewhat arbitrarily, $\mu(x,z) = \exp\big( (x^2 + z^2)/k^2\big)$, with  $k\in  \R$ a non-zero constant. The exponential of a formal variable is understood as a power series, so explicitly 
$$\mu(x,z) = \sum_{i = 0 }\frac{1}{n!} \left( \frac{z}{k}\right)^{2n} \, \exp \left( \frac{x^2}{k^2}\right).$$
The  $\Z_2^2$-manifold $\R^{1|1,1,1}$ can then be equipped with the warped Riemannian metric
$$g = \delta x\, \delta z + \delta \eta\, \delta \zx ~\exp \left( \frac{x^2+ z^2}{k^2}\right).$$
Note that  in the limit $k \rightarrow \infty $ the warping factor goes likes $\mu(x,z) \rightarrow 1$ and so we recover the canonical metric.
\end{example}

\subsection{Isometries and Killing vector fields}
The notions of isometries and infinitesimal isometries of a Riemannian $\Z_2^n$-manifold are the same as the classical notions.
\begin{definition}
Let $(M,g)$ and $(N,h)$ be Riemannian $\Z_2^n$-manifolds. Then a diffeomorphism  $\phi : M \rightarrow N$ is said to be an \emph{isometry} if and only if
$$\langle X | Y \rangle_g = \langle \phi_*X | \phi_*Y \rangle_h \,,$$
for all $X$ and $Y \in \Vect(M)$.
\end{definition}
\noindent \textbf{Observations:} As we are dealing with diffeomorphisms it is clear that $\dim(M) = \dim(N)$. Moreover, as $\deg(X) = \deg(\phi_* X)$ and similar for $Y$, it must be the case that $\deg(g) = \deg(h)$. Thus, it is impossible to have isometries  between Riemannian $\Z_2^n$-manifolds that have Riemannian metrics of different $\Z_2^n$-degrees.
\smallskip

Moving on to infinitesimal symmetries, any homogeneous  vector field $X \in \Vect(M)$ defines an infinitesimal local diffeomorphism, which in coordinates is given by
$$\phi^*_\lambda x^I = x^I + \lambda  \, X^I(x)\,,$$
where $\deg(\lambda) = \deg(X)$ is a formal parameter taken to $\Z_2^n$-commute with the coordinates.  This induces the map
$$\delta x^I \mapsto \delta x^I  + \lambda \, \delta x^J \, \frac{\partial X^I}{\partial x^J}\,.$$
It is then a direct calculation to show that  under such a local diffeomorphism  the metric changes as
\begin{align*}
g \mapsto & g + \lambda \, \delta x^J \delta x^I \left((-1)^{\langle \deg(X), \deg(I)\rangle} \, \frac{\partial X^K}{\partial x^J} g_{KI} + (-1)^{\langle \deg(X)+ \deg(J), \deg(I)\rangle} \, \frac{\partial X^K}{\partial x^I} g_{KJ} \right.\\
& \left. + (-1)^{\langle \deg(X), \deg(I) + \deg(J)\rangle}\, X^K \frac{\partial g_{IJ}}{\partial x^K} \right) + \cO(\lambda^2).
\end{align*}
The Lie derivative is then defined as the $\cO(\lambda)$ term of the above, i.e.,
$$ (L_Xg)_{IJ} =  (-1)^{\langle \deg(X), \deg(I)\rangle} \, \frac{\partial X^K}{\partial x^J} g_{KI} + (-1)^{\langle \deg(X)+ \deg(J), \deg(I)\rangle} \, \frac{\partial X^K}{\partial x^I} g_{KJ}
+ (-1)^{\langle \deg(X), \deg(I) + \deg(J)\rangle}\, X^K \frac{\partial g_{IJ}}{\partial x^K}. $$
Here, $L_X$ is the \emph{Lie derivative along $X$}, defined on the sheaf of tensors fields on $M$ analogously to the case of standard supermanifolds.
In particular, $L_X$ is the derivation of degree $(\deg(X),0)$ on the sheaf of covariant tensor fields (which is a sheaf of $\Z_2^n\times\mathbb{N}$-graded (tensor) algebras) such that $L_X(f)=X(f)$ and $L_X(\xd f)=\xd(X(f))$ for $f\in\cO_M$:
$$L_X(\za\ot\zb)=L_X(\za)\ot\zb+(-1)^{\langle \deg(X), \deg(\za)\rangle}\za\ot L_X(\zb)\,.$$
A Killing vector field is then a vector field whose Lie derivative annihilates the metric. More formally we have the following definition.
\begin{definition}
A vector field $X$ on a Riemannian $\Z_2^n$-manifold $(M,g)$ is said to be a \emph{Killing vector field} if and only if
$$L_Xg =0\,.$$
\end{definition}

Just as in the classical and super-case, we have the following result.
\begin{proposition}
Let $(M,g)$ be a Riemannian $\Z_2^n$-manifold. The space of Killing vector fields forms a $\Z_2^n$-Lie algebra under the standard $\Z_2^n$-graded Lie bracket of vector fields.
\end{proposition}
\begin{proof}
This follows directly from $L_{[X,Y]} = [L_X, L_Y]$ for all vector fields $X$ and $Y \in \Vect(M)$.  Indeed, this is clearly true for the action on 1-forms. Therefore, if $g=\za\ot\zb$ for 1-forms $\za$ and $\zb$, then
\beas&L_X\circ L_Y(g)=L_X\circ L_Y(\za\ot\zb)%=\left(L_X\circ L_Y-(-1)^{\langle \deg(X), \deg(Y)\rangle}L_Y\circ L_X\right)(\za\ot\zb)\\
=L_X\circ L_Y(\za)\ot\zb+(-1)^{\langle \deg(X), \deg(\za)+\deg(Y)\rangle}L_Y(\za)\ot L_X(\zb)\\
&-(-1)^{\langle \deg(Y), \deg(\za)\rangle}L_X(\za)\ot L_Y(\zb)+(-1)^{\langle \deg(Y)+\deg(X), \deg(\za)\rangle}(\za)\ot L_X\circ L_Y(\zb)
\eeas
and consequently
\beas&[L_X, L_Y](g)=\left(L_X\circ L_Y-(-1)^{\langle \deg(X), \deg(Y)\rangle}L_Y\circ L_X\right)(\za\ot\zb)\\
&=[L_X, L_Y](\za)\ot\zb+(-1)^{\langle \deg(Y)+\deg(X), \deg(\za)\rangle}\za\ot [L_X, L_Y](\zb)\\
&=L_{[X,Y]}(\za)\ot\zb+(-1)^{\langle \deg(Y)+\deg(X), \deg(\za)\rangle}\za\ot L_{[X, Y]}(\zb)=L_{[X, Y]}(g)\,.
\eeas
\end{proof}
\begin{example}
Consider a  $\Z_2^n$-manifold of the form $M = N \times \R^2$. Such a  $\Z_2^n$-manifold is a pp-wave type $\Z_2^n$-manifold if it is a Riemannian $\Z_2^n$-manifold with a degree zero Riemannian metric, such that in any coordinate system $(x^I, ~ s,t)$, where $(s,t)$ are canonical coordinates on $\R^2$, the Riemannian metric is of the form
$$g = \delta x^I \, \delta x^J g_{JI}(x) + 2 \, \delta s \, \delta t + \delta s \, \delta s \, H(x,s),$$
where $H : M \times \R \rightarrow \R_*$ is a smooth map of $\Z_2^n$manifolds, and $g_N = \delta x^I \, \delta x^J g_{JI}(x)$ is a degree zero Riemannian metric on $N$. It is easy to observe that there is a null Killing vector field $\frac{\partial}{\partial t}$, i.e., $\langle \partial_t | \partial_t \rangle_g = 0$ and $L_{\partial_t}g =0$. Such $\Z_2^n$-manifolds are generalisations of pp-wave type spacetimes as defined by Flores \& S\'{a}nchez \cite{Flores:2006}.
\end{example}
\subsection{The inverse metric}
The inverse metric  is locally defined via
$$g^{IK}g_{KJ} = g_{JK} g^{KI} = \delta_J^I.$$
Note that the inverse metric will have the same $\Z_2^n$-degree as the metric. The above relation can be used to deduce the symmetry property of the  inverse metric.
\begin{proposition}\label{prop:SymInv}
The inverse metric has the following symmetry
$$(-1)^{\langle \deg(K), \deg(K) \rangle } ~ g^{IK} = (-1)^{\langle\deg(I) , \deg(K) \rangle + \langle \deg(g), \deg(g)\rangle} ~ \big(  (-1)^{\langle \deg(I), \deg(I)\rangle} ~ g^{KI} \big).$$
\end{proposition}
\begin{proof}
\begin{align*}
g^{IK}g_{KJ} &= (-1)^{\langle \deg(I) + \deg(K) + \deg(g), \deg(K) + \deg(J) + \deg(g)   \rangle}  ~ g_{KJ}g^{IK}\\
& = (-1)^{\langle \deg(I) + \deg(K) + \deg(g), \deg(K) + \deg(J) + \deg(g)   \rangle + \langle \deg(K), \deg(J) \rangle + \phi }  ~ g_{JK}g^{KI}
\end{align*}
where $\phi$ represents the symmetry of the inverse metric, which is to be determined. Expanding the scalar product shows that it must be the case that
$\phi + \langle  \deg(I), \deg(J)\rangle + \langle \deg(I), \deg(g)  \rangle + \langle \deg(K), \deg(K)\rangle + \langle \deg(K), \deg(J)\rangle  + \langle \deg(g), \deg(J)  \rangle + \langle \deg(g) , \deg(g) \rangle =0$
when $I = J$. This means that  $\phi = \langle  \deg(J), \deg(J)  \rangle + \langle  \deg(K), \deg(K)   \rangle + \langle  \deg(K), \deg(J)  \rangle + \langle  \deg(g), \deg(g)  \rangle$, i.e.,
$$g^{JK} = (-1)^{  \langle  \deg(J), \deg(J)  \rangle + \langle  \deg(K), \deg(K)   \rangle + \langle  \deg(K), \deg(J)  \rangle + \langle  \deg(g), \deg(g)  \rangle} \, g^{KL}\,.$$
This directly implies the symmetry property stated in the proposition.
\end{proof}
\noindent \textbf{Observation:} If the Riemannian metric is even, then its corresponding inverse metric defines a symmetric contravariant rank-2 tensor. On the other hand, is the Riemannian metric is odd, then its corresponding inverse metric defines a skew-symmetric contravariant  rank-2 tensor.  This is in complete agreement with the case of even and odd Riemannian metrics on supermanifolds.
\smallskip

The inverse metric can be used to raise indices of covariant tensors and in particular, it is essential in defining the trace of rank-two covariant tensors. Specifically, starting with a rank-two covariant tensor $\omega_{JI}$, where we do not assume any symmetry, we can use  Lemma \ref{lem:VecOneform} and the inverse metric to define
$$X_{J}^{\:\:\: K} := (-1)^{\langle \deg(I), \deg(\omega) + \deg(g) + \deg(J) \rangle}\, \omega_{JI}g^{IK}\,.$$
The trace of $\omega_{JI}$ is then defined as the $\Z_2^n$-graded trace (see \cite[Theorem 1.]{Covolo:2012}) of $X_{J}^{\:\:\: K} $.
\begin{definition}\label{def:trace}
Let $(M,g)$ be a Riemannian $\Z_2^n$-manifold and let $\omega_{JI}$ be (the components of) a rank-two covariant tensor.  Then the \emph{trace of $\omega$} is defined as
$$\tr_g\big( \omega\big) :=  (-1)^{\langle \deg(\omega) + \deg(g), \deg(I) + \deg(J) \rangle} ~  (-1)^{\langle \deg(I), \deg(J) \rangle}\, \omega_{JI} \, g^{IJ}\, (-1)^{\langle \deg(J), \deg(J)\rangle } \in C^\infty(M)\,.$$
\end{definition}
Note that $\deg\left(\tr_g\big( \omega\big)\right) = \deg(\omega) + \deg(g)$. The formula for the trace of a rank-two tensor is identical to the classical formula up to the obligatory sign factors. We have chosen the second index in our definition, but this is sufficient as any covariant tensor can be written as the sum of its symmetrisation and skew-symmetrisation.%
\begin{lemma}\label{lem:trace}
Let $(M,g)$ be a Riemannian manifold. If $g$ is an odd Riemannian metric and $\omega_{JI}$ is symmetric, i.e., $\omega_{JI} = (-1)^{\langle \deg(I), \deg(J) \rangle}\, \omega_{IJ}$, then $\tr_g(\omega)=0$. Similarly, if $g$ is an even Riemannian metric and $\omega_{JI}$ is skew-symmetric, i.e., $\omega_{JI} = {-}(-1)^{\langle \deg(I), \deg(J) \rangle}\, \omega_{IJ}$, then $\tr_g(\omega)=0$.
\end{lemma}
\begin{proof}
Using Proposition \ref{prop:SymInv} we observe that in both the stipulated cases we contract tensors that are symmetric and skew-symmetric in their indices. Thus, $\tr_g(\omega) = {-} \tr_g(\omega) $ and so $\tr_g(\omega) =0$.
\end{proof}
\begin{example}
If $g$ is an odd Riemannian metric, then it is clear from Proposition \ref{prop:SymInv} and  Lemma \ref{lem:trace} that $\tr_g(g)=0$.
\end{example}

\subsection{Levi-Civita connections and the Fundamental Theorem of Riemannian $\Z_2^n$-geometry}
The notion of an affine connection on a $\Z_2^n$-manifold is more-or-less the same as that of an affine connection on a supermanifold, or indeed a manifold. In particular, we generalise Koszul's notion of a connection.
\begin{definition}\label{def:AffCon}
An \emph{affine connection} on a $\Z_2^n$-manifold is a $\Z_2^n$-degree preserving map
\begin{align*}
\nabla : &  ~ \Vect(M) \times \Vect(M) \longrightarrow \Vect(M)\\
 & (X,Y) \mapsto \nabla_X Y,
\end{align*}
that satisfies the following:
\begin{itemize}
\item Bi-linearity
\begin{align*}
& \nabla_X(Y+Z) = \nabla_X Y + \nabla_X Z,\\
&\nabla_{X +Y} Z = \nabla_X Z + \nabla_Y Z,
\end{align*}
\item $C^\infty(M)$-linearity in the first argument
$$\nabla_{fX} Y = f \, \nabla_X Y,$$
\item The Leibniz rule
$$\nabla_X fY = X(f) Y + (-1)^{\langle\deg(X), \deg(f) \rangle} ~ f \, \nabla_X Y,$$
\end{itemize}
for all  (homogeneous) $X,Y,Z \in \Vect(M)$ and $f \in C^\infty(M)$.
\end{definition}
We define the \emph{Christoffel symbols}  of any affine connection by (locally) setting
\begin{align*}
\nabla_X Y &= X^I \frac{\partial Y^J}{\partial x^I}\frac{\partial}{\partial x^J}  + (-1)^{\langle \deg(I), \deg(Y) + \deg(J) \rangle} ~ X^I Y^J \nabla_{\frac{\partial}{\partial x^I}} \left(\frac{\partial}{\partial x^J} \right)\\
&=  X^I \frac{\partial Y^J}{\partial x^I}\frac{\partial}{\partial x^J}  + (-1)^{\langle \deg(I), \deg(Y) + \deg(J) \rangle} ~ X^I Y^J \Gamma_{JI}^{\:\:\:K} \frac{\partial}{\partial x^K}.
\end{align*}
That is we make the standard definition
$$ \nabla_{\frac{\partial}{\partial x^I}} \left(\frac{\partial}{\partial x^J} \right) :=  \Gamma_{JI}^{\:\:\:K} \frac{\partial}{\partial x^K}.$$
\begin{definition}\label{def:SymCon}
The \emph{torsion tensor} of an affine connection  $T_\nabla : \Vect(M)\otimes_{C^\infty(M)} \Vect(M) \rightarrow \Vect(M)$   is defined as
$$T_\nabla(X,Y) := \nabla_X Y - (-1)^{\langle\deg(X), \deg(Y) \rangle} ~ \nabla_Y X - [X,Y],$$
 for any (homogeneous) $X, Y \in \Vect(M)$. An affine connection is said to be \emph{symmetric} if the torsion vanishes.
\end{definition}
It is clear that the torsion is skew-symmetric, i.e.,
$$ T_\nabla(X,Y) = -  (-1)^{\langle\deg(X), \deg(Y) \rangle} ~ T_\nabla(Y,X).$$
Furthermore, a direct calculation shows that
$$T_\nabla(fX,Y) = f \,T_\nabla(X,Y) = (-1)^{\langle \deg(f), \deg(X) \rangle}~ T_\nabla(X,fY),$$
i.e., we do indeed have a tensor.  The components of the torsion tensor with respect local coordinates are
$$T_{IJ}^{\:\:\:K} = \Gamma_{JI}^{\:\:\:K} -(-1)^{\langle \deg(I), \deg(J) \rangle}  \Gamma_{IJ}^{\:\:\:K}.$$
\begin{definition}\label{def:MetComp}
An affine connection on a Riemannian $\Z_2^n$-manifold $(M,g)$ is said to be \emph{metric compatible} if and only if
$$X\langle Y|Z \rangle_g  = \langle \nabla_X Y|Z \rangle_g  + (-1)^{\langle \deg(X), \deg(Y) \rangle} ~ \langle Y|\nabla_X Z \rangle_g\,,$$
for any $X, Y, Z \in \Vect(M)$.
\end{definition}
\begin{lemma}\label{lem:KosFor}
Let $(M,g)$ be a Riemannian $\Z_2^n$-manifold. Then an affine connection $\nabla$ that is symmetric (i.e., is torsionless) and metric compatible satisfies the Koszul formula
\begin{align}\label{eqn:KosFor}
2 \langle \nabla_X Y | Z  \rangle_g & = X\langle Y | Z   \rangle_g  + \langle [X,Y]| Z \rangle_g \\ \nonumber
 & +(-1)^{\langle \deg(X), \deg(Y) + \deg(Z) \rangle}~ \big( Y \langle Z|X\rangle_g - \langle [Y,Z]| X \rangle _g \big)\\ \nonumber
 &-(-1)^{\langle \deg(Z), \deg(X) + \deg(Y) \rangle}~ \big( Z \langle X|Y\rangle_g - \langle [Z,X]| Y \rangle _g \big),
\end{align}
for all homogeneous $X,Y$ and $Z \in \Vect(M)$.
\end{lemma}
\begin{proof}
The  Koszul formula \eqref{eqn:KosFor} can be constructed in the same way as the classical case upon being careful with the needed sign factors.  From the metric compatibility (Definition \ref{def:MetComp}) we have
\begin{align*}
& X\langle Y|Z \rangle_g  = \langle \nabla_X Y|Z \rangle_g  + (-1)^{\langle \deg(X), \deg(Y) \rangle} ~ \langle Y|\nabla_X Z \rangle_g\,,\\
& Y\langle Z|X \rangle_g  = \langle \nabla_Y Z|X \rangle_g  + (-1)^{\langle \deg(Y), \deg(Z) \rangle} ~ \langle Z|\nabla_Y X \rangle_g\,,\\
& Z\langle X|Y \rangle_g  = \langle \nabla_Z X|Y \rangle_g  + (-1)^{\langle \deg(Z), \deg(X) \rangle} ~ \langle X|\nabla_Z Y \rangle_g\,.\\
\end{align*}
We then sum these three with appropriate sign factors  to get (after a little rearranging)
\begin{align*}
& X\langle Y|Z \rangle_g + (-1)^{\langle \deg(X), \deg(Y) + \deg(Z) \rangle} ~ Y\langle Z|X \rangle_g  {-} (-1)^{\langle \deg(Z), \deg(X) + \deg(Y) \rangle} Z\langle X|Y \rangle_g\\
& = \langle \nabla_X Y|Z \rangle_g + (-1)^{\langle \deg(X), \deg(Y) \rangle} ~ \langle \nabla_Y X |Z \rangle_g\\
& + (-1)^{\langle \deg(X), \deg(Y) \rangle}~ \langle Y | \nabla_X Z - (-1)^{\langle \deg(X), \deg(Z) \rangle} ~ \nabla_Z X  \rangle_g\\
&+  \langle X | \nabla_Y Z - (-1)^{\langle \deg(Z), \deg(Y) \rangle} ~ \nabla_Z Y  \rangle_g.
\end{align*}
Next we use the symmetric property (Definition \ref{def:SymCon}) the RHS of the above becomes
\begin{align*}
&= 2 \langle \nabla_X Y |Z \rangle_g - \langle[X,Y] | Z \rangle_g + (-1)^{\langle \deg(X), \deg(Y) \rangle}~ \langle Y | [X,Z] \rangle_g +  \langle X | [Y,Z] \rangle_g\\
&=2 \langle \nabla_X Y |Z \rangle_g - \langle[X,Y] | Z \rangle_g - (-1)^{\langle \deg(Z), \deg(X) + \deg(Y) \rangle} ~ \langle [Z,X]|Y \rangle_g + (-1)^{\langle \deg(X), \deg(Y) + \deg(Z) \rangle} ~ \langle [Y,Z]|X \rangle_g.
\end{align*}
Rearranging this produces the Koszul formula.
\end{proof}
We are now in  a position to state and prove  the generalisation of the fundamental theorem of Riemannian geometry to $\Z_2^n$-geometry.  That is, we have a canonical affine connection associated with the Riemannian structure.  
\begin{theorem}[The Fundamental Theorem]\label{trm:FundTheo}
There is a unique symmetric (torsionless) and metric compatible affine connection on a Riemannian $\Z_2^n$-manifold $(M,g)$.
\end{theorem}
\begin{proof}
The classical proof using the Koszul formula generalises to our setting. Lemma \ref{lem:KosFor} shows that if an affine connection is metric compatible and symmetric then it necessarily satisfies the Koszul formula \eqref{eqn:KosFor}.  We need to argue that the Koszul formula defines a unique  metric compatible and symmetric connection. With this in mind, let us fix $X$  and $Y \in \Vect(M)$ (both homogeneous) and then consider the RHS of \eqref{eqn:KosFor} as a map $\omega(Z) : \Vect(M) \rightarrow C^\infty(M)$. A direct calculation shown that $\omega(f \,Z) = (-1)^{\langle \deg(f) , \deg(X)+ \deg(Y) \rangle}~ f \, \omega(Z)$. Thus, $\omega$ is a tensor, and more specifically a one-form. We can then use Lemma \ref{lem:VecOneform}. We  define the corresponding vector field as $2 \nabla_X Y$ so that $\omega(Z) = \langle 2 \nabla_X Y, Z\rangle_g$ for all vector fields $Z$. In this way we obtain a map $\nabla : \Vect(M) \times \Vect(M) \rightarrow \Vect(M)$. We need to show that this map does indeed define an affine connection (see Definition \ref{def:AffCon}).
\begin{itemize}
\setlength\itemsep{1em}
\item Bi-linearity is clear as the RHS of the Koszul formula \eqref{eqn:KosFor} is bi-linear.
\item The $C^\infty(M)$-linearity in the first argument follows from a direct calculation:
\begin{align*}
\langle 2 \nabla_{fX}Y | Z \rangle_g & = f X\langle Y|Z \rangle_g + \langle [fX,Y]|Z \rangle_g\\
& + (-1)^{\langle \deg(X) + \deg(f), \deg(Y) + \deg(Z) \rangle} \big(Y \langle Z|fX\rangle_g - \langle [Y,Z]| fX \rangle_g \big)\\
  & - (-1)^{\langle\deg(Z), \deg(X) + \deg(f) + \deg(Z) \rangle} \big( Z \langle fX| Y\rangle_g   -   \langle [Z,fX]| Y\rangle_g\big)\\
  & = f \left ( X\langle Y|Z \rangle_g + \langle [X,Y]|Z \rangle_g \right.\\
  & + (-1)^{\langle \deg(X), \deg(Y) + \deg(Z) \rangle } \big(Y \langle Z|X\rangle_g - \langle [Y,Z]| X \rangle_g \big)\\
  & - \left.  (-1)^{\langle \deg(Z), \deg(X) + \deg(Y) \rangle } \big(Z \langle X|Y\rangle_g - \langle [Z,X]| Y \rangle_g \big) \right)\\
  & - (-1)^{\langle \deg(X), \deg(f) , \deg{Y}\rangle}\, Y(f)\langle X|Z\rangle_g + (-1)^{\langle \deg(X), \deg(f) , \deg{Y}\rangle + \langle \deg{X}, \deg(Z) \rangle}\, Y(f)\langle Z|X\rangle_g\\
  & - (-1)^{\langle \deg(Z), \deg(X)+ \deg(f) + \deg(Y) \rangle}\, Z(f) \langle X|Y\rangle_g + (-1)^{\langle \deg(Z), \deg(X)+ \deg(f) + \deg(Y) \rangle}\, Z(f) \langle X|Y\rangle_g \\
   & = \langle 2 \, f\, \nabla_X Y|Z \rangle_g.
\end{align*}
This is precisely the $C^\infty(M)$-linearity property.
\item The Leibniz rule follows from a direct calculation:
\begin{align*}
\langle 2 \nabla_X fY\rangle_g & = X\langle fY| Z \rangle_g + \langle [X, fY]|Z  \rangle_g\\
& + (-1)^{\langle \deg(X), \deg(Y) + \deg(f) + \deg(Z) \rangle}\, \big( f Y \langle Z |X\rangle_g - \langle [fY ,Z]| X \rangle_g \big)\\
& + (-1)^{\langle \deg(Z), \deg(Y) + \deg(f) + \deg(X) \rangle}\, \big( Z \langle X |fY\rangle_g - \langle [Z,X]| fY \rangle_g \big)\\
& = X(f) \langle Y |Z\rangle_g + (-1)^{\langle \deg(X), \deg(f) \rangle}\, f X\langle Y |Z\rangle_g\\
& + X(f) \langle Y |Z\rangle_g + (-1)^{\langle \deg(X), \deg(f) \rangle}\,  f \langle [X,Y]| Z\rangle_g\\
&+ (-1)^{\langle \deg(X), \deg(f) \rangle + \langle \deg(X), \deg(Y) + \deg(Z) \rangle}\, f\big(Y \langle Z|X \rangle_g   - \langle[Y,Z]|X  \rangle_g\big)\\
&+ (-1)^{\langle \deg(X), \deg(f) \rangle + \langle \deg(X), \deg(Y) \rangle + \langle \deg(Z), \deg(X)+ \deg(Y) + \deg(f) \rangle}\, Z(f) \langle Y|X \rangle_g\\
&-(-1)^{\langle \deg(X), \deg(f) \rangle + \langle \deg(Z) ,\deg(X) + \deg(Y)\rangle}\, f \big( Z\langle X|Y\rangle_g  - \langle [Z,X]|Y\rangle_g\big)\\
&-(-1)^{\langle \deg(X), \deg(f)  \rangle + \langle \deg(X) + \deg(Y) + \deg(f) \rangle}\, Z(f) \langle X|Y\rangle_g \\
& = (-1)^{\langle(X), \langle(f)}\left(  \big(fX \langle Y|Z\rangle_g  + \langle f[X,Y]|Z \rangle_g \big) \right.\\
&  + (-1)^{\langle \deg(X), \deg(Y) + \deg(Z)\rangle}\,\big(fY \langle Z|X\rangle_g  + \langle f[Y,Z]|X \rangle_g \big) \\
& \left.  - (-1)^{\langle \deg(Z), \deg(X) + \deg(Y)\rangle}\,\big(fZ \langle X|Y\rangle_g  + \langle f[Z,X]|Y \rangle_g \big) \right) + \langle 2 X(f)Y|Z \rangle_g\\
& = \langle 2 X(f)Y + (-1)^{\langle \deg(X), \deg(f)\rangle} \, f \,2 \nabla_X Y | Z  \rangle_g.
\end{align*}
This is precisely the Leibniz rule.
\end{itemize}
Thus, the Koszul formula defines a connection $\nabla$. We can then use the  Lemma \ref{lem:KosFor} to conclude that this connection is indeed metric compatible and symmetric.
\end{proof}
\begin{definition}
The symmetric, metric compatible affine connection defined by   Theorem \ref{trm:FundTheo} is referred to as the \emph{Levi-Civita connection}.
\end{definition}
\noindent \textbf{Observation:} The Fundamental Theorem holds on any Riemannian $\Z_2^n$-manifold irrespectively of the $\Z_2^n$-degree of the Riemannian metric. That is, any Riemannian $\Z_2^n$-manifold canonically comes equipped with the Levi-Civita connection. This encompasses the well-known result for Riemannian supermanifolds. 

\begin{proposition}\label{prop:ChriSymLoc}
In  local coordinates $x^I$  the Christoffel symbols $\Gamma_{JI}^{\:\:\: L}$ of the Levi-Civita connection on a Riemannian $\Z_2^n$-manifold $(M,g)$ are given by
$$\Gamma_{JI}^{\:\:\: L} = \frac{1}{2}\left(\frac{\partial g_{JK}}{\partial x^I} + (-1)^{\langle \deg(I), \deg(J)\rangle} \, \frac{\partial g_{IK}}{\partial x^J} - (-1)^{\langle \deg(K), \deg(I) + \deg(J) \rangle} \, \frac{\partial g_{IJ}}{\partial x^K} \right)g^{KL}\,.$$
\end{proposition}
\begin{proof}
It follows from direct computation using the Koszul formula, see Lemma \ref{lem:KosFor} and Theorem \ref{trm:FundTheo}. In particular, it is easy to see that
$$2 \langle \nabla_I \partial_J | \partial_K\rangle_g = 2\,  \Gamma_{JI}^{\:\:\: L}\, g_{LK} = \left(\frac{\partial g_{JK}}{\partial x^I} + (-1)^{\langle \deg(I), \deg(J)\rangle} \, \frac{\partial g_{IK}}{\partial x^J} - (-1)^{\langle \deg(K), \deg(I) + \deg(J) \rangle} \, \frac{\partial g_{IJ}}{\partial x^K} \right).$$
Multiplying on the right of the above by the inverse metric yields the result.
\end{proof}

\subsection{The  Riemann curvature tensor}
We now turn attention to the generalisation of the Riemann curvature tensor and its basic properties.
\begin{definition}\label{def:RiemCurTen}
The \emph{Riemannian curvature tensor} of an affine connection
$$R_\nabla : \Vect(M)\otimes_{C^\infty(M)}\Vect(M)  \otimes_{C^\infty(M)} \Vect(M) \rightarrow \Vect(M)$$
is defined as
$$R_\nabla(X,Y)Z = [\nabla_X, \nabla_Y]Z - \nabla_{[X,Y]}Z,$$
for all $X,Y$ and $Z \in \Vect(M)$.
\end{definition}
Directly from the definition it is clear that
\begin{equation}\label{eqn:AntiSymRieTens}
R_\nabla(X,Y)Z = {-}(-1)^{\langle \deg(X), \deg(Y)\rangle} \, R_\nabla(Y,X)Z,
\end{equation}
for all $X, Y$ and $Z\in \Vect(M)$.
\begin{proposition}\label{prop:FirstBianc}
The Riemannian curvature tensor associated with the Levi-Civita connection satisfies the $\Z_2^n$-graded first Bianchi identity
$$(-1)^{\langle\deg(X), \deg(Z)  \rangle} \, R_\nabla (X,Y) Z  + (-1)^{\langle\deg(Y), \deg(X)  \rangle} \, R_\nabla (Y,Z) X  + (-1)^{\langle\deg(Z), \deg(Y)  \rangle} \, R_\nabla (Z,X) Y  =0\,,$$
for any and all $X,Y$ and $Z \in \Vect(M)$.
\end{proposition}
\begin{proof}
The proof follows in almost identically the same way as the classical proof upon a little care with the sign factors and as such we just sketch the proof here. It is more convenient to re-write the LHS of of the first Bianchi identity as
$$Bi(X,Y,Z) := R_\nabla (X,Y) Z  + (-1)^{\langle\deg(Y), \deg(X) + \deg(Z)  \rangle} \, R_\nabla (Y,Z) X  + (-1)^{\langle\deg(Z), \deg(Y) + \deg(X)  \rangle} \, R_\nabla (Z,X) Y.$$
Then using the definition of the Riemann curvature tensor (Definition \ref{def:RiemCurTen}) we obtain
\begin{align*}
&Bi(X,Y,Z)  = \nabla_X (\nabla_Y Z - (-1)^{\langle \deg(Z), \deg(Y)\rangle} \, \nabla_Z Y) - \nabla_{[X,Y]}Z\\
& + (-1)^{\langle \deg(X), \deg(Y)+ \deg(Z)\rangle}\,\nabla_Y (\nabla_Z X - (-1)^{\langle \deg(Z), \deg(X)\rangle} \, \nabla_X Z) - (-1)^{\langle \deg(X), \deg(Y)+ \deg(Z)\rangle}\, \nabla_{[X,Y]}Z\\
&+ (-1)^{\langle \deg(Z), \deg(X)+ \deg(Y)\rangle}\,\nabla_Z (\nabla_X Y - (-1)^{\langle \deg(X), \deg(Y)\rangle} \, \nabla_Y X) - (-1)^{\langle \deg(Z), \deg(X)+ \deg(Y)\rangle}\, \nabla_{[Z,X]}Y,\\
& \textnormal{then using the fact that the Levi-Civita connection is torsionless (Definition \ref{def:SymCon}),}\\
&= \nabla_X [Y,Z] - \nabla_{[X,Y]}Z\\
& +(-1)^{\langle \deg(X), \deg(Y)+ \deg(Z)\rangle} \,\nabla_Y [Z,X] -(-1)^{\langle \deg(X), \deg(Y)+ \deg(Z)\rangle} \, \nabla_{[Y,Z]}X\\
& +(-1)^{\langle \deg(Z), \deg(X)+ \deg(Y)\rangle} \,\nabla_Z [X,Y] -(-1)^{\langle \deg(Z), \deg(X)+ \deg(Y)\rangle} \, \nabla_{[Z,X]}Y\\
&\textnormal{once again using the fact that the Levi-Civita connection is torsionless,}\\
&=  [X, [Y,Z] -[[X,Y],Z] -(-1)^{\langle \deg(X), \deg(Y) \rangle  } \, [Y, [X,Z]].
\end{align*}
This is nothing but the $\Z_2^n$-graded Jacobi identity for the standard Lie bracket of vector fields (here written in Loday--Leibniz form). Thus, $Bi(X,Y,Z) =0$ and so the  $\Z_2^n$-graded first Bianchi identity holds.
\end{proof}
\begin{remark}
Just as in the classical case, the proof of the  first Bianchi identity  only relies on the affine connection being symmetric, i.e., the torsion is zero. Thus, the first Bianchi identity holds for any symmetric affine connection and not just the Levi-Civita connection.
\end{remark}
\begin{proposition}
The  Riemann curvature tensor associated with the Levi-Civita connection satisfies
$$\langle R_\nabla(X,Y)Z| W \rangle_g = - (-1)^{\langle \deg(Z), \deg(W)\rangle} \, \langle R_\nabla(X,Y)W| Z \rangle_g\,,$$
for all $X,Y, Z$ and $W \in \Vect(M)$.
\end{proposition}
\begin{proof}
The proposition can be established via direct calculation along the same lines as the classical case.  First, directly from the definition of the Riemann curvature tensor and using the metric compatibility condition we observe that
\begin{align*}
\langle R_\nabla(X,Y)Z| W \rangle_g &=  \langle \nabla_X \nabla_Y Z |W  \rangle_g - (-1)^{\langle \deg(X), \deg(Y) \rangle} \,  \langle \nabla_Y \nabla_X Z |W  \rangle_g  - \langle \nabla_{[X,Y]}Z |W \rangle_g\\
&= X\langle \nabla_Y Z|W  \rangle_g -  (-1)^{\langle \deg(X), \deg(Y)+ \deg(Z) \rangle}\, \langle \nabla_Y Z | \nabla_X W \rangle_g\\
& - (-1)^{\langle \deg(X), \deg(Y)\rangle} \, Y \langle \nabla_X Z | W \rangle_g + (-1)^{\langle \deg(Y), \deg(Z)\rangle} \, \langle \nabla_X Z| \nabla_Y W \rangle_g - \langle \nabla_{[X,Y]}Z|W  \rangle_g.
\end{align*}
Then, after a little algebra
\begin{align*}
&\langle R_\nabla(X,Y)Z| W \rangle_g + (-1)^{\langle \deg(Z), \deg(W)\rangle} \, \langle R_\nabla(X,Y)W| Z \rangle_g\\
&  = X \langle \nabla_Y Z|W\rangle_g + (-1)^{\langle \deg(Z), \deg(Y)\rangle} \, X \langle Z| \nabla_Y W \rangle_g\\
&- (-1)^{\langle \deg(X), \deg(Y)\rangle}\, Y\langle \nabla_X Z  | W \rangle_g - (-1)^{\langle \deg(X) , \deg(Y) + \deg(Z) \rangle }\, Y \langle Z | \nabla_X W  \rangle_g\\
& - \langle \nabla_{[X,Y]}Z|W \rangle_g - (-1)^{\langle\deg(Z), \deg(X) + \deg(Y) \rangle}\, \langle Z | \nabla_{[X,Y]} \rangle_g - (-1)^{\langle \deg(Z), \deg(X) + \deg(Y)\rangle}\, \langle Z | \nabla_{[X,Y]}W \rangle_g,\\
& \textnormal{then using metric compatibility again}\\
& = \big ( XY - (-1)^{\langle\deg(X), \deg(Y) \rangle} \, Y Z - [X,Y]\big)\, \langle Z | W \rangle_g =0.
\end{align*}
\end{proof}

The covariant derivative of the Riemann curvature tensor is defined via the product rule for the covariant derivative, i.e.,
\begin{align*}
\big (  \nabla_Z R_\nabla\big )(X,Y)W &=  \nabla_Z \big ( R_\nabla(X,Y)W\big ) -  R_\nabla( \nabla_Z X,Y)W  \\
&- (-1)^{\deg(Z), \deg(X)} R_\nabla(X,\nabla_Z Y)W - (-1)^{\deg(Z), \deg(X)+ \deg(Y)} R_\nabla(X, Y)\nabla_Z W\,
\end{align*}
for all $X,Y,Z$ (homogeneous) and $W \in \Vect(M)$.
\begin{proposition}\label{prop:SecondBianc}
The  Riemann curvature tensor associated with the Levi-Civita connection satisfies the $\Z_2^n$-graded second Bianchi identity
$$(-1)^{\langle \deg(Y), \deg(Z)\rangle}\, \big (\nabla_Z   R_{\nabla}\big)(X,Y) + (-1)^{\langle \deg(X), \deg(Y)\rangle}\, \big (\nabla_Y   R_{\nabla}\big)(Z,X) + (-1)^{\langle \deg(Z), \deg(X)\rangle}\, \big (\nabla_X   R_{\nabla}\big)(Y,Z) =0.$$
\end{proposition}
\begin{proof}
The proposition  follow from a direct computation along the same lines as the classical calculation. We will only outline the proof here. We denote  the LHS of the second Bianchi identity as $Bii(X,Y,Z)W$, for arbitrary   $X,Y,Z$ (homogeneous) and $W \in \Vect(M)$.  Using the fact that the Levi-Civita connection is torsionless, we can, after a little algebra, write
$$
Bii(X,Y,Z)W  = \sum_{cyclic} (-1)^{\langle \deg(X), \deg(Z) \rangle}\, \big(  (\nabla_X R_\nabla)(Y,Z)W + R_\nabla(X,Y)\nabla_Z W - R_\nabla([X,Y],Z)W\big).
$$
Then using the definition of the Reimannian curvature tensor, we obtain an expression involving one, two and three $\nabla$'s.  However, all but the terms involving a single $\nabla$ cancel in the cyclic sum.  One then obtains
$$
 Bii(X,Y,Z)W  = \sum_{cyclic}(-1)^{\langle \deg(X), \deg(Z) \rangle}\, \nabla_{[[X,Y],Z]}W,
$$
and then we recognise the Jacobi identity for the Lie bracket of vector fields and so $Bii(X,Y,Z) =0$.
\end{proof}
\begin{proposition}
In  local coordinate $x^I$  the components of the Riemann curvature tensor, which are defined by $R_{\nabla}(\partial_I, \partial_J)\partial_K := R_{IJK}^{\:\:\:\:\:\:\: L}\partial_L$, of the Levi-Civita connection on a Riemannian $\Z_2^n$-manifold $(M,g)$ are given by
\begin{align*}
R_{IJK}^{\:\:\:\:\:\:\: L} & =  \frac{\partial \Gamma_{KJ}^{\:\:\: \:\: L}}{\partial x^I} - (-1)^{\langle \deg(I), \deg(J) \rangle} \, \frac{\partial \Gamma_{KI}^{\:\:\: \:\: L}}{\partial x^J}\\ & + (-1)^{\langle \deg(I), \deg(J) + \deg(K) + \deg(M)\rangle}\, \Gamma_{KJ}^{\:\:\: M}\Gamma_{M I}^{\:\:\: \:\: L}\\
& - (-1)^{\langle \deg(J),  \deg(K) + \deg(M)\rangle}\, \Gamma_{KI}^{\:\:\: M}\Gamma_{M J}^{\:\:\: \:\: L}\,.
\end{align*}
\end{proposition}
\begin{proof}
This follows from direct computation.
\end{proof}

\subsection{The Ricci tensor and Ricci scalar}
The Ricci curvature tensor and Ricci scalar can also be generalised  to the setting of Riemannian $\Z_2^n$-manifolds.
\begin{definition}\label{def:RicCurv}
The \emph{Ricci curvature tensor} of an affine connection is the symmetric rank-2 covariant tensor defined as
$$Ric_\nabla(X,Y) = \tr\left( Z \mapsto \half \big ( R_\nabla(Z,X)Y + (-1)^{\langle \deg(X), \deg(Y)  \rangle} \,R_\nabla(Z,Y)X \big) \right),$$
for all $X$ and $Y \in \Vect(M)$. Here $\tr$ is the $\Z_2^n$-graded trace (see \cite[Theorem 1.]{Covolo:2012}).
\end{definition}
\begin{remark}
 One can show that the tensor  $\tr\left( Z \mapsto  R_\nabla(X,Y)Z  \right)$ vanishes if the Riemannian metric is even, but not necessarily so if the Riemannian metric is odd. Combining this with the first Bianchi identity (see Proposition \ref{prop:FirstBianc}) one can prove that $ \tr\left( Z \mapsto  R_\nabla(Z,X)Y  \right)$ is not symmetric if the Riemannian metric is odd. Thus,  we define the Ricci curvature tensor as the symmetrisation of a trace in order to obtain a tensor with a well-defined symmetry for both even and odd Riemannian metrics.
\end{remark}
Defining $R_{IJ} := Ric_\nabla(\partial_I, \partial_J)$, it is clear that the symmetry equates to $R_{IJ} = (-1)^{\langle \deg(I), \deg(J) \rangle} R_{JI}$.
\begin{definition}\label{def:RicScal}
Let $(M,g)$ be a Riemannian $\Z_2^n$-manifold. The \emph{Ricci scalar} (associated with the Levi-Civita connection) is defined as
$$S_\nabla := \tr_g\big(Ric_\nabla \big) \in C^\infty(M),$$
where $\tr_g$ is the metric trace (see Definition \ref{def:trace}).
\end{definition}
Note that $\deg(S_\nabla) = \deg(g)$, and that this has some direct consequences.
\begin{proposition}
Let $(M,g)$ be a Riemannian $\Z_2^n$-manifold, then the Ricci scalar $S_\nabla$ cannot be a  non-zero constant if the metric $g$ has non-zero $\Z_2^n$-degree.
 \end{proposition}
\begin{proof}
It is sufficient to consider this question locally and so we can employ local coordinates. If the metric is of non-zero $\Z_2^n$-degree, then a non-zero Ricci scalar must locally be at least linear in one or more of the formal coordinates. Thus it cannot be constant, i.e., locally the derivative  of the Ricci scalar with respect to at least one of the formal coordinates is non-vanishing. The only possible constant of non-zero $\Z_2^n$-degree is zero.
\end{proof}
\noindent \textbf{Observation:} The only possible Riemannian $\Z_2^n$-manifolds with constant, but non-zero, Ricci scalar are those whose Riemannian metric is degree zero.
\begin{proposition}\label{prop:OddRVan}
Let $(M,g)$ be a Riemannian $\Z_2^n$-manifold with an odd metric. Then the Ricci scalar vanishes identically.
\end{proposition}
\begin{proof}
In any set of local coordinates, we see that
$$S_\nabla = (-1)^{\langle  \deg(g), \deg(I) + \deg(J)\rangle}\, R_{JI}g^{IJ}\,(-1)^{\langle  \deg(J), \deg(J)\rangle}.$$
Using the fact that the Ricci tensor is symmetric in its indices, while the inverse metric is skew-symmetric for an odd Riemannian metric (see Proposition \ref{prop:SymInv}),  the contraction of the two will vanish.
\end{proof}
\subsection{The covariant divergence and the connection Laplacian}
The notions of the gradient of a function, the covariant divergence of a vector field and the connection Laplacian all generalise to the setting of $\Z_2^n$-manifolds. We restrict attention to the covariant divergence and the connection Laplacian (with respect to the Levi-Civita connection)  to avoid subtleties of introducing Berezin volumes on $\Z_2^n$-manifolds (see \cite{Poncin:2016}). The reader should note that unless the Riemannian metric is of degree zero, there is no canonical Berezin volume. Thus, in order not to introduce extra structure, one should consider the connection Laplacian and not the  Laplace--Beltrami operator.   
\begin{definition}
Let $f \in C^\infty(M)$ be an arbitrary function on a Riemannian $\Z_2^n$-manifold $(M,g)$. The \emph{gradient of $f$} is the unique vector field $\textrm{grad}_g f$ such that 
$$X(f) = (-1)^{\langle \deg(f), \deg(g) \rangle }\, \langle X | \textrm{grad}_g f \rangle_g ,$$
for all  $X \in \Vect(M)$.
\end{definition}
In local coordinates we have
$$\textrm{grad}_g f = (-1)^{\langle \deg(f), \deg(g) \rangle + \langle \deg(f)+ \deg(g), \deg(J) \rangle } \, \frac{\partial f}{\partial x^J}g^{JI} \frac{\partial}{\partial x^I}.$$
Note that $\deg(\textrm{grad}_g f) = \deg(f) + \deg(g)$. It is easy to observe that
\begin{equation}\label{eqn:GradProd}
\textrm{grad}_g (f f') = (-1)^{\langle\deg(f), \deg(g) \rangle}\,f \, \textrm{grad}_g f' + (-1)^{\langle \deg(f'), \deg(f) + \deg(g)\rangle} \,  f ' \, \textrm{grad}_g f\,,
\end{equation}
for homogeneous $f$ and $f' \in C^\infty(M)$. Moreover, it is clear from the definition of the gradient that if $f = \Id$, i.e., the constant function of value one, then its gradient is zero.
\begin{definition}
Let $(M,g)$ be a Riemannian $\Z_2^n$-manifold and let $\nabla$ be the associated Levi-Civita connection. The \emph{covariant divergence} is the map
$$\textrm{Div}_\nabla : \Vect(M) \rightarrow C^\infty(M),$$
given by
$$ \textrm{Div}_\nabla X := \tr \big (Z \mapsto \nabla_Z X  \big),$$
for any arbitrary $X \in \Vect(M)$.
\end{definition}
Locally, the above definition amounts to 
$$ \textrm{Div}_\nabla X = (-1)^{\langle \deg(I), \deg(I)+ \deg(X) \rangle}\, \big( \nabla_I X\big)^I =  (-1)^{\langle \deg(I), \deg(I)+ \deg(X) \rangle}\, \frac{\partial X^I}{\partial x^I} + (-1)^{\langle \deg(I), \deg(I)+ \deg(J) \rangle}\, X^J \Gamma_{J I}^{\:\:\: \:\: I}.$$
The covariant divergence satisfies the expected properties (it is straightforward to check these)
\begin{align}
&\textrm{Div}_\nabla (X + c Y) = \textrm{Div}_\nabla X + c \, \textrm{Div}_\nabla Y , \textnormal{and}\\
& \label{eqn:DivfX} \textrm{Div}_\nabla (fX) = (-1)^{\langle \deg(f), \deg(X)+ \deg(g) \rangle}\, \langle X | \textrm{grad}_g f \rangle_g + f \, \textrm{Div}_\nabla X,    
\end{align}
with $X$ and $Y \in \Vect(M)$ and $c \in \R$, and and $f \in C^\infty(M)$. The connection Laplacian is then defined as ``div of the grad". More carefully we have the following definition.
\begin{definition}
Let $(M,g)$ be a Riemannian $\Z_2^n$-manifold and let $\nabla$ be the associated Levi-Civita connection.  The \emph{connection Laplacian} (acting on functions) is the differential operator of $\Z_2^n$-degree $\deg(g)$ defined as
$$\Delta_g f =  \textrm{Div}_\nabla( \textrm{grad}_g f ), $$
for any and all $f \in C^\infty(M)$.
\end{definition}
Directly from the definition of the gradient, one can see that the connection Laplacian has no zeroth-order term, i.e., $\Delta_g \Id =0$.  By construction, it is obvious that the connection Laplacian is at most second-order.  The connection Laplacian on a Riemannian $\Z_2^n$-manifold is exactly what one expects given the classical connection Laplacian on a Riemannian manifold.   
\begin{proposition}\label{prop:ConLapVioDer}
Let $(M,g)$ be Riemannian $\Z_2^n$-manifold with an even Riemannian metric. Furthermore, let $\Delta_g$ be the connection Laplacian associated with the Levi-Civita connection. Then the violation of the Leibniz rule over the product $f f'$ for $\Delta_g$ is given by  $(-1)^{\langle \deg(f'), \deg(g)\rangle}\,  2\,\langle\textrm{grad}_g f | \textrm{grad}_g f' \rangle_g\,$, where  $f$ and $f'$ are  homogeneous but otherwise arbitrary functions.
\end{proposition}
\begin{proof} 
We prove the proposition by  constructing the `anomaly' to the Leibniz rule for the connection Laplacian.  Let $(M,g)$ be a Riemannian $\Z_2^n$-manifold where the  Riemannian metric is of arbitrary $\Z_2^n$-degree (we will not insist that it is even at this stage). Using \eqref{eqn:GradProd}   and \eqref{eqn:DivfX} we obtain
\begin{align*}
\Delta_g(ff')& = \textrm{Div}_\nabla \left( (-1)^{\langle \deg(f), \deg(g) \rangle} \, f \, \textrm{grad}_g f' + (-1)^{\langle \deg(f'), \deg(f)+\deg(g) \rangle} \, f' \, \textrm{grad}_g f \right)\\
& = (-1)^{\langle \deg(f), \deg(f')+ \deg(g)\rangle} \, \langle \textrm{grad}_g f' | \textrm{grad}_g f  \rangle_g + (-1)^{\langle \deg(f), \deg(g)\rangle}\, f (\Delta_g f')\\
& + (-1)^{\langle \deg(f'), \deg(g) \rangle }\, \langle \textrm{grad}_g f | \textrm{grad}_g f'  \rangle_g +(-1)^{\langle \deg(f') , \deg(f)+ \deg(g) \rangle}\,  f'(\Delta_g f)\\
&= (\Delta_g f) f' + (-1)^{\langle \deg(g), \deg(f) \rangle} \, f (\Delta_g f')\\
& + (-1)^{\langle \deg(f'), \deg(g)\rangle}\left(1 + (-1)^{\langle\deg(g), \deg(g) \rangle} \right) \, \langle\textrm{grad}_g f | \textrm{grad}_g f' \rangle_g\,,
\end{align*}
with $f$ and $f' \in C^\infty(M)$. Then, if the Riemannian metric is even the `anomaly' is precisely as claimed in the proposition. 
\end{proof}
The proof of Proposition \ref{prop:ConLapVioDer} implies that for odd Riemannian metrics the connection Laplacian is a vector field.  However, as we shall show, the connection Laplacian associated with an odd Riemannian metric vanishes. Before we can prove this we need a lemma.
\begin{lemma}\label{lem:ContLC}
Let $(M,g)$ be a Riemannian $\Z_2^n$-manifold and let $\Gamma_{J I}^{\:\:\: \:\: K}$ be the associated Christoffel symbols in some chosen coordinate system. Then
\begin{align}\label{eqn:ContLC} 
\nonumber (-1)^{\langle \deg(I), \deg(I)+ \deg(L)\rangle} \, 2 \, g^{JL}\Gamma_{J I}^{\:\:\: \:\: I} &= (-1)^{\langle \deg(I), \deg(I) \rangle} \, g^{JI}\left(\frac{\partial g_{IM}}{\partial x^L} \right)g^{MI}\\
& + (-1)^{\langle \deg(I), \deg(I)+ \deg(L)\rangle} \, (1-(-1)^{\langle \deg(g), \deg(g) \rangle}) \, g^{JL}\left( \frac{\partial g_{LM}}{\partial x^I} \right)g^{MI}\,.
\end{align}
In particular, if the Riemannian metric is odd, then 
\begin{equation}\label{eqn:oddContLC}
 g^{JL}\Gamma_{J I}^{\:\:\: \:\: I} =  g^{JL}\left( \frac{\partial g_{LM}}{\partial x^I} \right)g^{MI}\,.
\end{equation}
\end{lemma}
\begin{proof}
From Proposition \ref{prop:ChriSymLoc} we observe that
\begin{align*}
(-1)^{\langle \deg(I), \deg(I)+ \deg(L)\rangle} \, 2 \,\Gamma_{LI}^{\:\:\: I} &= (-1)^{\langle \deg(I), \deg(I)+ \deg(L)\rangle}\, \left(\frac{\partial g_{LM}}{\partial x^I}\right) g^{MI} \\
&    + (-1)^{\langle \deg(I), \deg(I)\rangle} \, \left(\frac{\partial g_{IM}}{\partial x^L}\right) g^{MI}\\
& - (-1)^{ \langle \deg(M), \deg(I)+ \deg(L) \rangle + \langle \deg(I), \deg(I) + \deg(L) \rangle} \, \left(\frac{\partial g_{IL}}{\partial x^M}\right)g^{IM} \,.
\end{align*}
Relabelling the contracted indices of the  last term by $I \rightarrow M$ and $M \rightarrow I$, and then using the symmetry of the Riemannian metric and its inverse (see Proposition \ref{prop:SymInv})  we see that 
\begin{align*}  
&(-1)^{ \langle \deg(M), \deg(I)+ \deg(L) \rangle + \langle \deg(I), \deg(I) + \deg(L) \rangle} \, \left(\frac{\partial g_{IL}}{\partial x^M}\right)g^{IM}\\ & =  (-1)^{ \langle \deg(I), \deg(I)+ \deg(L) \rangle + \langle \deg(g), \deg(g) \rangle} \, \left(\frac{\partial g_{LM}}{\partial x^I}\right)g^{IM}\,.
\end{align*}
This establishes the validity of \eqref{eqn:ContLC}. If the Riemannian metric is odd, then the first term of \eqref{eqn:ContLC} vanishes as this involves the contraction of a symmetric and a skew-symmetric object, while the second term remains. This establishes \eqref{eqn:oddContLC}.
\end{proof}
\begin{theorem}\label{trm:OddLapVanish}
The connection Laplacian on a Riemannian $\Z_2^n$-manifold $(M,g)$ with an odd Riemannian metric identically vanishes, i.e., it is a zero map.
\end{theorem}
\begin{proof}
A direct computation from the definitions and using 
$$\frac{\partial g^{KI}}{\partial x^I} = - (-1)^{\langle\deg(I), \deg(K) + \deg(J)+ \deg(g) \rangle} \, g^{KJ}\left(\frac{\partial g_{JL}}{\partial x^I} \right)g^{LI}\,,$$
one can quickly show that in local coordinates
\begin{align*}
\Delta_g f &= (-1)^{\langle \deg(J), \deg(J)\rangle }\, g^{JI}\frac{\partial^2 f}{\partial x^I \partial x^J}\\
& - (-1)^{\langle \deg(I), \deg(I) + \deg(L) \rangle + \langle \deg(J) , \deg(J)\rangle } \, g^{JL}\left(\frac{\partial g_{LM}}{\partial x^I} \right)g^{MI} \frac{\partial f}{\partial x^J}\\
&+(-1)^{\langle \deg(I), \deg(I) + \deg(L) \rangle + \langle \deg(J) , \deg(J)\rangle } \, g^{JL}\Gamma_{L I}^{\:\:\: \:\: I}\frac{\partial f}{\partial x^J}\,.
\end{align*}
It is clear from the proof of Proposition \ref{prop:ConLapVioDer}, or from the fact that the contraction of a symmetric object and a skew-symmetric object, that the second-order term vanishes for odd Riemannian metrics.  Then using Lemma \ref{lem:ContLC}, and specifically \eqref{eqn:oddContLC}, shows that $\Delta_g f=0$ for all $f \in C^\infty(M)$ provided the Riemannian metric is odd.
\end{proof}
\begin{remark}
 The reader should  be reminded of modular vector fields and the modular class of a Poisson manifold (see \cite{Weinstein:1997}).  Via analogy, Theorem \ref{trm:OddLapVanish}  is not surprising given that the modular class of a symplectic manifold vanishes, or in other words, Hamiltonian vector fields annihilate the Liouville volume.  However, there is no direct generalisation of the modular class for supermanifolds or $\Z_2^n$-manifolds equipped with odd Riemannian metrics.  
\end{remark}
\noindent \textbf{Observation:} One cannot develop a theory of harmonic functions etc. on  Riemannian $\Z_2^n$-manifolds with odd Riemannian metrics using the canonical connection Laplacian.\par 
Just for completeness, the connection Laplacian on  Riemannian $\Z_2^n$-manifold with an even Riemannian metric is locally given by
$$\Delta_g f = (-1)^{\langle \deg(J), \deg(J) \rangle } \, g^{JI}\left( \frac{\partial^2 f}{\partial x^I \partial x^J} - \Gamma_{I J}^{\:\:\: \:\: K} \, \frac{\partial f}{\partial x^K} \right),$$ 
which is, up to the sign factors, in agreement with the classical connection Laplacian on a Riemannian manifold (see for example \cite[Section 2.6, page 30]{Andrews:2011}). Deriving this local expression follows using the same steps as in the classical case (so Lemma \ref{lem:ContLC} and Proposition \ref{prop:ChriSymLoc} together with the symmetry of the metric and its inverse) and so details are left to the reader.

\section{Discussion}\label{sec:Disc}
In this paper we have defined a Riemannian $\Z_2^n$-manifold and explored their basic properties, such as the Fundamental Theorem, i.e.,  a canonical symmetric metric compatible affine connection, the Levi-Civita connection, always exists irrespective of the degree of the Riemannian metric.  We have deduced the $\Z_2^n$-graded generalisations of several objects of interest in Riemannian geometry such as the torsion and various curvature tensors associated with the Levi-Civita connection. In particular, it was shown that Riemannian $\Z_2^n$-manifolds with an odd metric, i.e., the total degree is 1, necessarily have vanishing scalar curvature. Thus, the Ricci scalar is not a very interesting curvature invariant in this case.  A little more generally, the scalar curvature cannot be a non-zero constant unless the Riemannian metric is of $\Z_2^n$-degree zero.   Thus, studying Riemannian $\Z_2^n$-manifolds with non-zero constant Ricci scalar is limited not just to the even case, but to the degree zero case.  Recall that de Sitter (dS) and anti-de Sitter (AdS) are Lorentzian Riemannian manifolds with constant positive and constant negative scalar curvature, respectively. Both these manifolds are of significance in cosmology and quantum field theory (the AdS/CFT correspondence, for example).  We have a higher graded ``no-go theorem''.
\medskip

\noindent \textbf{Statement:} \emph{There is no direct analogue of dS or AdS space in $\Z_2^n$-geometry where the Riemannian metric is of non-zero $\Z_2^n$-degree.}

\medskip

Four-dimensional Riemannian Einstein manifolds are important in the study of gravitational instantons and  Lorentzian Einstein manifolds play a r\^ole in string theory, quantum gravity and can serve as the target space of various sigma-models. A natural question here is if we have generalisations of Einstein manifolds were we replace the manifold with a Riemannian $\Z_2^n$-manifold with an arbitrary (but homogeneous) degree metric. The obvious thing to try is to set
$$Ric_\nabla = \kappa \, g,$$
where $Ric_\nabla$ is the Ricci curvature tensor associated with the Levi-Civita connection. It is clear that we require $\deg(\kappa) = \deg(g)$. Thus, $\kappa \in C^\infty(M)$ cannot be a non-zero constant if $\deg(g) \neq 0$.  We have another higher graded ``no-go theorem''.

\medskip

\noindent \textbf{Statement:} \emph{There is no direct analogue of an Einstein manifold in $\Z_2^n$-geometry where the Riemannian metric is of non-zero $\Z_2^n$-degree.}

\medskip

In conclusion, the basic notions, concepts, and mathematical statements found in classical Riemannian geometry, with a little care, generalise to the setting of $\Z_2^n$-geometry.  The challenge is to now find interesting examples and possible applications of Riemannian $\Z_2^n$-geometry in physics and classical differential geometry.

\section*{Acknowledgements}
A.J.~Bruce cordially thanks Norbert Poncin for many fruitful discussions about $\Z_2^n$-manifolds and related subjects.   J.~Grabowski acknowledges that his  research was funded by the  Polish National Science Centre grant HARMONIA under the contract number 2016/22/M/ST1/00542.

\end{document}